\documentclass[journal]{IEEEtran}
\usepackage{url}            
\usepackage{booktabs}       
\usepackage{amsfonts}       
\usepackage{nicefrac}       
\usepackage{microtype}      
\usepackage{lipsum}
\usepackage{amsmath}
\usepackage{graphicx}
\usepackage{amssymb}
\usepackage{pdfpages}
\usepackage{float}
\usepackage{subfiles} 
\usepackage{amsthm}
\usepackage{mathtools}
\usepackage{mathrsfs}
\usepackage{subfig} 
\newtheorem{theorem}{Theorem}
\newtheorem{example}{Example}

\newtheorem{lemma}{Lemma}

\newtheorem{remark}{Remark}

\newtheorem{assumption}{Assumption}
\newtheorem{problem}{Problem}
\DeclareUnicodeCharacter{221E}{-}

\usepackage{algorithm}
\usepackage{algpseudocode}
\newcommand\ignore[1]{{}}
\algrenewcommand\algorithmicrequire{\textbf{Input:}}
\algrenewcommand\algorithmicensure{\textbf{Output:}}
\usepackage{siunitx} 
\DeclareMathOperator{\diag}{diag}
\usepackage{multirow}
\usepackage{pdfpages}
\usepackage{cite}
\begin{document}
\author{Cong Li$^{\dagger}$,~\IEEEmembership{}
        Fangzhou Liu{$^{\dagger}$$^*$}~\IEEEmembership{Member,~IEEE},
        Yongchao Wang,~\IEEEmembership{}
        and~Martin Buss~\IEEEmembership{Fellow,~IEEE} 
\thanks{$^{\dagger}$ Common first authors; $^*$ Corresponding author}
\thanks{This work was supported in part by the National Postdoctoral Program of China under Grants (GZC20233538).}
\thanks{C. Li is with the College of Intelligence Science and Technology, National University of Defense Technology, Changsha 410073, China. e-mail: lc@nudt.edu.cn.}
\thanks{Y. Wang, and M.Buss are with the Chair of Automatic Control Engineering, Technical University of Munich, Theresienstr. 90, 80333, Munich, Germany. e-mail: \{yongchao.wang, mb\}@tum.de.}
\thanks{F. Liu is with the Research Institute of Intelligent Control and Systems, Harbin Institute of Technology, 150001,  Harbin, China, e-mail: fangzhou.liu@hit.edu.cn.}
\newline
}
\title{Data Informed Residual Reinforcement Learning for High-Dimensional Robotic Tracking Control}
\maketitle
\begin{abstract}
The learning inefficiency of reinforcement learning (RL) from scratch hinders its practical application towards continuous robotic tracking control, especially for high-dimensional robots. This work proposes a data-informed residual reinforcement learning (DR-RL) based robotic tracking control scheme applicable to robots with high dimensionality. The proposed DR-RL methodology outperforms common RL methods regarding sample efficiency and scalability. Specifically, we first decouple the original robot into low-dimensional robotic subsystems; and further utilize one-step backward (OSBK) data to construct incremental subsystems that are equivalent model-free representations of the above decoupled robotic subsystems. The formulated incremental subsystems allow for parallel learning to relieve computation load and offer us mathematical descriptions of robotic movements for conducting theoretical analysis. Then, we apply DR-RL to learn the tracking control policy, a combination of incremental base policy and incremental residual policy, under a parallel learning architecture. The incremental residual policy uses the guidance from the incremental base policy as the learning initialization and further learns from interactions with environments to endow the tracking control policy with adaptability towards dynamically changing environments. Our proposed DR-RL based tracking control scheme is developed with rigorous theoretical analysis of system stability and weight convergence. The effectiveness of our proposed method is validated numerically on a 7-DoF KUKA iiwa robot manipulator and experimentally on a 3-DoF robot manipulator that would fail for other counterpart RL methods.
\end{abstract}
\IEEEpeerreviewmaketitle

\section{Introduction}
Reinforcement learning (RL)  holds the promise of autonomous learning of continuous robotic control policies adaptable to varying environmental conditions \cite{sutton2018reinforcement}.
However, RL from scratch requires an extensive amount of training data (interactions with surrounding environments) before learning one satisfying policy \cite{kim2021dynamic,ma2023position}. 
This sample inefficiency inhibits RL from robotic practical applications.
For example, large amounts of physical interactions between the robot and the real-world environments would cause undesirable mechanical wear and even damage to the robot itself and the environment.
The high dimensionality of robots would exacerbate the sample inefficiency problem mentioned above.
Thereby, this work leverages the mechanism of model-based RL (MBRL) and residual RL (RRL) to facilitate the data-efficient training of continuous robotic control policies.
Furthermore, the training is conducted in a parallel learning architecture, wherein the required computational load is distributed into multiple processors to relieve the computation complexity of high-dimensional robotic control tasks.
\subsection{Related Works}
The sample inefficiency causes the inapplicability of RL from scratch to robots given samples are often expensive to get. 
Reducing sample complexity motivates the advancement of MBRL and RRL fields.
The common MBRL improves sample efficiency by firstly learning a latent model of environment dynamics and then using the learned latent model to simulate experience for policy learning \cite{sutton1991dyna,polydoros2017survey,eysenbach2022mismatched}. 
Although the latent models learned via computation-intensive parametric \cite{janner2019trust,yao2022model} or nonparametric \cite{deisenroth2013gaussian,deisenroth2011pilco} methods contribute to reduced amounts of environmental samples, the model learning process introduces additional computation complexity.
Besides, the learned latent models are often black-box input-out mappings that inhibit designers from conducting further theoretical analysis. 
How to efficiently learn a control-oriented latent model favoring both computation simplicity and theoretical analysis remains an open problem in the MBRL field.
This motivates us to utilize the so-called one-step backward (OSBK) data to inform explicit latent models in simple incremental mathematical forms that would facilitate the theoretical analysis and ease the model learning difficulty.
From a different perspective, RRL enhances data efficiency by learning upon one base policy and further optimizing performance via the residual policy trained by RL, rather than learning the solution from the very beginning \cite{johannink2019residual}.
The guidance from the base policy constrains the search space of the residual policy. This improves the exploration efficiency of the residual policy learning process.
Thereby, the amount of training data and the consumed time is substantially reduced before a satisfying policy is learned.
The current RRL-related works \cite{johannink2019residual,silver2018residual,alakuijala2021residual,kulkarni2022learning,rana2023residual,zhang2022residual} are featured with task generality; however, the theoretical completeness remains to be further investigated.
Therefore, this work designs the residual policy part from a control-theoretical RL perspective in favor of theoretical analysis.

The control-theoretical RL, namely approximate dynamic programming (ADP), is one RL branch featured with available theoretical analysis, 
which serves as one promising candidate that contributes to RRL with theoretical completeness.
Although applicable for addressing the robotic (sub)optimal tracking control problem, ADP-based approaches \cite{modares2014optimal,kamalapurkar2015approximate,zhang2011data} suffer learning inefficiency regarding different desired trajectories.
In particular, the tracking control policy trained on one specific reference trajectory dynamics cannot efficiently track the unaccounted reference signals. 
However, reference signals are often described by different trajectory dynamics in complex tasks.
Thereby, the associated training process would repeatedly restart but might not satisfy the required tracking performance in each learning period; This is unfavorable to practical real-time applications.
For example, in a dynamic environment populated with moving obstacles, a robot keeps on replanning to generate safe desired trajectories, which are accounted for by multiple different trajectory dynamics. 
Thereby, the controller training processes in the works above \cite{kamalapurkar2015approximate,zhang2011data,modares2014optimal} would repeatedly restart as replanning happens. 
However, the tracking performance during each limited training period is usually not satisfying for practical applications, especially considering safety issues demanding perfect tracking precision.
This work utilizes the mechanism of control-theoretical RL to design the residual policy with a preference for theoretical completeness; 
While an assumed reference trajectory dynamics is avoided in the learning process. 
Thereby, the flexibility of our proposed tracking control scheme is extended.

In view of dimensionality, the works \cite{kamalapurkar2015approximate,zhang2011data,modares2014optimal} discussed above are mostly restricted to low-dimensional domains due to learning inefficiency.
This problem becomes even worse when referring to the robotic control with high dimensionality.
A common approach in MBRL and RRL fields to solving high-dimensional control tasks is using the powerful approximation ability of deep neural networks (DNNs) to learn the associated policies and/or latent models \cite{janner2019trust,plaat2020deep}.
However, extensive amounts of samples are still required for training DNNs, which would negate the benefit of sample efficiency brought by MBRL and RRL. 
The control-theoretical RL methods \cite{kamalapurkar2015approximate,zhang2011data,modares2014optimal} also suffer limited scalability towards high-dimensional robotic control tasks.
The obstacle lies in the so-called curse of complexity. 
Specifically, the required number of activation functions to gain a sufficiently accurate approximation of a value function grows exponentially with the system dimension \cite{kamalapurkar2016efficient}.
Even though a suitable large set of activation functions and appropriate hyperparameters are found through tedious engineering efforts, the accompanying computation load would 
degrade the realtime performance of the associated weight update law and the learned policy \cite{zhao2020robust,yang2022adaptive}. 
Thus, experimental validations of control-theoretical RL-based (sub)optimal tracking control policy on a high-dimensional robot are seldom found in existing works.
This work relieves the high sample complexity and computation load involved in high-dimensional robotic control tasks via a parallel learning architecture.
In particular, parallel learners learn solutions to decomposed sub-problems independently while working toward a common goal.

\subsection{Contributions}
The contributions of our work are summarized as follows.
\begin{itemize}
\item 
A data-efficient and scalable DR-RL based robotic tracking control scheme is proposed to be applicable to high-dimensional robots, 
which succeeds in experimental tasks where common RL methods are intractable.
\item The formulated data-informed incremental subsystems 
 offer latent models for the learning process to improve sample efficiency;
 present a mathematical description of the robotic movement for theoretical analysis;
 and allow for the application of parallel learning architecture to relieve computation complexity.
\item 
The proofs of system stability and weight convergence are provided on the basis of the formulated incremental subsystems and the parallel learners utilized off-policy experience data.
\end{itemize}

The organization of this paper is as follows. 
Section \ref{section problem formulation} presents the problem formulation. 
Then, the development of incremental subsystems is shown in Section \ref{sec incremental subsystem}. 
Thereafter, Section \ref{section optimal tracking control} presents the mechanism of the proposed robotic tracking control scheme.
Section \ref{sec Approximate solutions} elucidates the approximate solution learned in parallel. 
The developed robotic tracking control policy is numerically and experimentally validated in Section \ref{sec simulation} and Section \ref{section Experimental validation}, respectively.
Finally, the conclusion is drawn in Section \ref{section conclusion}.

\emph{Notations:} Throughout this paper, 
$\mathbb{R}$ denotes the set of real numbers; 
$\mathbb{R}^{n}$ is the Euclidean space of $n$-dimensional real vector; 
$\mathbb{R}^{n \times m}$ is the Euclidean space of $n \times m$ real matrices; 
$\left\|\cdot\right\|$ represents the Euclidean norm for vectors and induced norm for matrices;
The pseudo-inverse of the full column rank $D$ is denoted as $D^{\dagger} :=(D^{\top}D)^{-1}D^{\top} \in \mathbb{R}^{m \times n}$;
$\diag {[x]}$ is the $n \times n$ diagonal matrix with the $i$th diagonal entry equals $x_i$.
For notational brevity, time dependence is suppressed without causing ambiguity. 

\section{Problem Formulation} \label{section problem formulation}
The robotic movement is assumed to be described by the control-affine nonlinear system:
\begin{equation} \label{eq general sys}
\dot{x} = f(x)+g(x)u(x),
\end{equation}
where $x \in \mathbb{R}^{n}$, $u(x) : \mathbb{R}^{n} \to \mathbb{R}^{m}$ are the system state and control input, respectively. 
Both $f(x) : \mathbb{R}^{n} \to \mathbb{R}^{n}$ and
$g(x) : \mathbb{R}^{n} \to \mathbb{R}^{n \times m}$ are bounded and locally Lipschitz.
Assume that the explicit knowledge of $f(x)$ and $g(x)$ is unknown and $rank(g) = n$ holds for the robot \eqref{eq general sys}.

This work focuses on the high-dimensional robotic tracking control task presented in Problem \ref{Problem 1}.
\begin{problem} \label{Problem 1}
Given a desired trajectory $x_d \in \mathbb{R}^{n}$,
learning an efficient tracking control policy $u(x)$ applicable to the high-dimensional robot \eqref{eq general sys}.
\end{problem}
The high-uncertainty property of Problem \ref{Problem 1} encourages us to use RL based approaches. 
However, the high-dimensionality property of Problem \ref{Problem 1}
invalidates the basic RL from scratch approaches with limited scalability. 
In the following sections, we clarify our scalable and efficient RL methodology to solve Problem \ref{Problem 1}.

\section{Data Informed Incremental Subsystem} \label{sec incremental subsystem}
This section benefits from the decoupling technique and the OSBK data \cite{hsia1989new} to develop incremental subsystems. 
The formulated incremental subsystems jointly describe the movement of the original robot \eqref{eq general sys} without using explicit model information.
Specifically, the high-dimensional robot is first decoupled into multiple low-dimensional subsystems in Section \ref{sec Decoupled Subsystem}. 
Then, the OSBK data is utilized to estimate the unknown model information (including coupling terms between subsystems) for constructing the incremental subsystems in Section \ref{sec devp incre sub}.

\subsection{Decoupled Subsystem} \label{sec Decoupled Subsystem}
The high-dimensional robot \eqref{eq general sys} is supposed to be decoupled into $N \in \mathbb{R}^{+}$ subsystems, wherein the $i$th subsystem follows
\begin{equation} \label{i-th state space system}
    \dot{x}_i = f_i + g_{i} u_i, \ \  i = 1,2,\cdots,N,
\end{equation} 
where $x_i \in \mathbb{R}^{n_{i}}$, 
    $u_i \in \mathbb{R}^{m_{i}}$ are the  local state and control input of the $i$th subsystem;
$f_i \in \mathbb{R}^{n_{i}}$ is a combination of the local internal dynamics and  the coupling  terms of the $i$th subsystem.
$g_i \in \mathbb{R}^{n_i \times m_i}$ is the local  input gain matrix.
Note that both $f_i$ and $g_i$ are state-dependent matrices.
Besides, $n = \sum_{i}^{N} n_i$ and  $m = \sum_{i}^{N} m_i$ hold.

For a better explanation, we focus on the robot manipulator case in Example \ref{example 1} to show the transformation from the robot \eqref{eq general sys} into its associated decoupled subsystems \eqref{i-th state space system}.

\begin{example} \label{example 1}
 The robot manipulator could be described by the Euler-Lagrange (E-L) equation \cite{li2021concurrent}:
\begin{equation} \label{example E-L}
     M(q)\Ddot{q} + N(q,\dot{q})+F(\dot{q})=\tau, 
\end{equation}
where $M(q): \mathbb{R}^{n_r}\to  \mathbb{R}^{n_r \times n_r}$ is the symmetric positive definite inertia matrix; 
$N(q,\dot{q}) := C(q,\dot{q})\dot{q}+G(q): \mathbb{R}^{n_r} \times \mathbb{R}^{n_r} \to \mathbb{R}^{n_r}$, $C(q,\dot{q}):  \mathbb{R}^{n_r} \times\mathbb{R}^{n_r} \to \mathbb{R}^{n_r \times n_r}$ is the matrix of centrifugal and Coriolis terms, $G(q): \mathbb{R}^{n_r} \to \mathbb{R}^{n_r}$ represents the gravitational terms; $F(\dot{q}): \mathbb{R}^{n_r}\to \mathbb{R}^{n_r}$ denotes the viscous friction; $q$, $\dot{q}$, $\Ddot{q}\in \mathbb{R}^{n_r}$ are the vectors of angles, velocities, and accelerations, respectively; $\tau \in \mathbb{R}^{m_r}$ represents the input torque vector.
A fully actuated robot manipulator is considered here, thus $n_r = m_r$.

The high-dimensional robot \eqref{example E-L} could be decoupled into $n_r$ subsystems, wherein the $i$th subsystem reads
\begin{equation} \label{i-th E-L}
     M_{ii}\Ddot{q}_i + \sum_{j=1,j\ne i}^{n} M_{ij}\Ddot{q}_{j}+N_i+F_i = \tau_i, \ \  i = 1,2,\cdots,n_r,
\end{equation}
where $M_{ij}$ denotes the $ij$th entry of the matrix $M$, and $N_{i}$ ($F_{i}$)  is the $i$th entry of the  vector $N$ ($F$). 

Let $x_{i_{r_{1}}} := q_i \in \mathbb{R}$, 
$x_{i_{r_{2}}} :=\dot{q}_i \in \mathbb{R}$, 
$f_{i_{r}} := - (\sum_{j=1,j\ne i}^{n} M_{ij}\Ddot{q}_{j}+N_i+F_i)/M_{ii} \in \mathbb{R}$, 
and $g_{i_{r}} :=  1/M_{ii}  \in \mathbb{R}$, 
we rewrite \eqref{i-th E-L} as
\begin{subequations} \label{example i-th state space system}
     \begin{align}
     & \dot{x}_{i_{r_{1}}} = x_{i_{r_{2}}}, \label{example i-th state space system 1}\\
     & \dot{x}_{i_{r_{2}}} = f_{i_{r}} + g_{i_{r}}  \tau_i  \label{example i-th state space system 2},\ \  i = 1,2,\cdots,n_r
     \end{align}
\end{subequations}
\end{example}
\begin{remark}
The decoupled subsystems \eqref{i-th state space system} are beneficial to alleviate the computation complexity induced by the high dimensionality property of Problem \ref{Problem 1}.
This is because the decoupled subsystems allow for the parallel learning architecture in Section \ref{sec Approximate solutions}, wherein the required intensive computational load for solving Problem \ref{Problem 1} is distributed into multiple processors.
\end{remark}
\subsection{Incremental Subsystem} \label{sec devp incre sub}
This subsection exploits the OSBK data to estimate the unknown model knowledge $f_i$ and $g_{i}$ in \eqref{i-th state space system}. 
This departs from common methods that identify the unknown $f_i$, $g_{i}$ explicitly through a tedious identification process \cite{kamalapurkar2016modeltracking,liu2021high,beckers2017stable, na2017adaptive}.

To facilitate estimation, we first introduce a predetermined constant matrix $\bar{g}_{i}  \in \mathbb{R}^{n_i \times m_i}$ and multiply $\bar{g}^{\dagger}_{i}$ on \eqref{i-th state space system},
\begin{equation} \label{incrmental system 1}
\begin{aligned}
\bar{g}^{\dagger}_{i} \dot{x}_{i} = h_i + u_i,
\end{aligned}
\end{equation}
where $h_i := (\bar{g}^{\dagger}_{i} - g^{\dagger}_{i})\dot{x}_{i} + g^{\dagger}_{i} f_i \in \mathbb{R}^{m_{i}}$ is a lumped term that embodies all of the unknown model knowledge in \eqref{i-th state space system}.

Then, with a sufficiently high sampling rate  
\footnote{The so-called  sufficiently high sampling rate, which is a prerequisite for estimating the unknown $h_i$ by reusing past measurements of states and control inputs, can be chosen as the value that is larger than $30$ times the system bandwidth \cite{li2021model}.} 
\cite{youcef1992input,li2021model}, we estimate the unknown $h_i$ as 
\begin{equation} \label{incrmental system 2}
\begin{aligned}
 \hat{h}_i   = h_{i,0}  = \bar{g}^{\dagger}_{i} \dot{x}_{i_{0}}-u_{i,0},
\end{aligned}
\end{equation}
utilising the OSBK data $\dot{x}_{i,{0}} = \dot{x}_{i}(t-t_s)$,  $u_{i,0} =u_i(t-t_s)$, where $t_s \in \mathbb{R}^{+}$ is the sampling time.

Substituting \eqref{incrmental system 2} into \eqref{incrmental system 1}, we finally obtain the $i$th incremental subsystem
\begin{equation} \label{incrmental system}
    \dot{x}_{i} =\dot{x}_{i,0} + \bar{g}_i (\Delta u_i+\xi_{i}),
\end{equation} 
where  $\Delta u_i := u_i - u_{i,0} \in \mathbb{R}^{m_{i}}$ is the incremental policy;
and $\xi_{i} := h_i- \hat{h}_i \in \mathbb{R}^{m_{i}}$ denotes the estimation error proved to be bounded in Section \ref{section optimal tracking control} \emph{Lemma \ref{boud of TDE AE}}.

The above formulated incremental subsystem \eqref{incrmental system} is an equivalent of the $i$th subsystem \eqref{i-th state space system}; however, no explicit model information is required.
\begin{remark}
The OSBK data ($\dot{x}_{i,0}$ and $u_{i,0}$ in particular) informed  \eqref{incrmental system} offers us with one model-free representation of the subsystem \eqref{i-th state space system}.
The resulting incremental subsystem benefits both efficient, parallel implementation and rigorous theoretical analysis. 
In particular, the informed incremental subsystem allows us to conduct the value function learning process in Section \ref{sec Approximate solutions} following the MBRL mechanism in a parallel manner. 
Thereby, the learning efficiency is substantially improved. 
Furthermore, the combination of the incremental subsystems offers a mathematical form describing the robotic movement, which permits us to use the tool from the control field to conduct the rigorous theoretical analysis.
\end{remark} 

Through the aforementioned analysis \eqref{i-th state space system}-\eqref{incrmental system}, we could decompose the robotic tracking task in Problem \ref{Problem 1} into sub-tasks regarding the incremental subsystems \eqref{incrmental system}, as clarified in Problem \ref{Problem 2}.
\begin{problem} \label{Problem 2}
    For the incremental subsystem \eqref{incrmental system}, learning the incremental policy $\Delta u_i$ that 
    drives the subsystem to track its associated desired trajectory precisely.
\end{problem}
In the following section, we focus on Problem \ref{Problem 2} to present our proposed tracking control scheme.

\section{DR-RL Based Robotic Tracking Control Scheme} \label{section optimal tracking control}
This section details our proposed DR-RL based robotic tracking control scheme (see Figure.~\ref{algorithm framework}), wherein the incremental policies for solving Problem \ref{Problem 2} are learned in parallel.

Under a parallel learning architecture, the learning process follows the RRL mechanism, wherein the incremental policy
\begin{equation} \label{incremental control}
\Delta u_{i} = \Delta u_{i_{b}}+\Delta u_{i_{r}},
\end{equation}
is an addition of the incremental base policy  $\Delta u_{i_{b}} \in \mathbb{R}^{m_{i}}$ 
and the incremental residual policy $\Delta u_{i_{r}} \in \mathbb{R}^{m_{i}}$. 
The detailed procedures to design $\Delta u_{i_{b}}$, $\Delta u_{i_{r}}$ and also their roles are later clarified in Section \ref{sec base policy} and Section \ref{sec residual control}, respectively.

Let $e_i := x_i - x_{d_{i}} \in \mathbb{R}^{n_{i}}$ denote the local tracking error, where $x_{d_{i}} \in \mathbb{R}^{n_{i}}$ denotes the local desired trajectory of the $i$th subsystem. 
Combining with \eqref{incrmental system} and \eqref{incremental control}, we would get the incremental error dynamics
\begin{equation} \label{eq err dyn}
    \dot{e}_i =  \dot{x}_{i,0} + \bar{g}_{i} (\Delta u_{i_{b}} + \Delta u_{i_{r}} + \xi_{i})-\dot{x}_{d_{i}}.
\end{equation}
In the remaining part of this section, we will focus on \eqref{eq err dyn} to clarify the designed incremental base and residual policies that jointly enforce the local tracking error $e_i$ to zero.
\begin{figure*}[!t]
\centering
\includegraphics[width=7in]{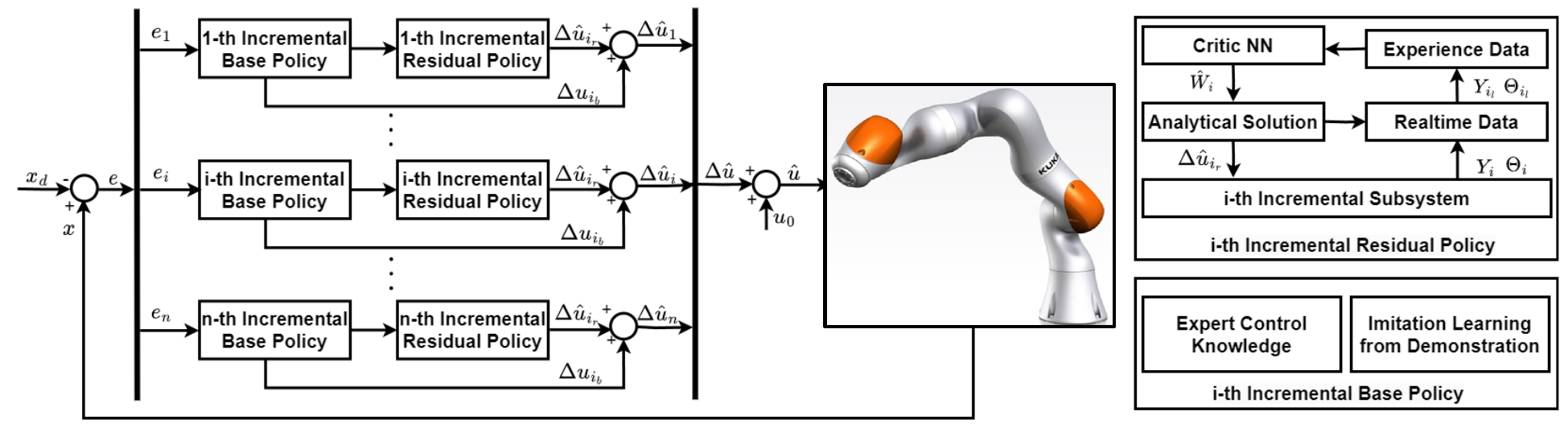} 
\caption{Schematic of the DR-RL based robotic tracking control policy.
The original high-dimensional robotic tracking task is decoupled into subtasks of incremental subsystems for efficient parallel implementation. 
The incremental policies, a combination of the incremental base policy and the incremental residual policy, are learned in parallel to solve the associated subtasks. 
The incremental base policy provides a policy initialization for the subsequent incremental residual policy learning process.
}
\label{algorithm framework}
\end{figure*}

\subsection{Incremental Base Policy} \label{sec base policy}
This subsection details the design of the incremental base policy that would offer guidance for the incremental residual policy learning process clarified in Section \ref{sec residual control}.
The guidance would simplify the exploration.
Thus, the learning difficulty of the incremental residual policy for complex and longer-horizon tasks is decreased.

Practitioners could use existing knowledge from either control or learning fields to design the incremental base policy, which is presented in detail in the following.
\subsubsection{Control Perspective}
The incremental base policy could be implemented as fine-tuned feedforward or feedback controllers such as proportional-integral-derivative, impedance control, or nonlinear dynamic inversion (NDI).
In this case, the existing control knowledge is utilized to guide the learning process 
and offers us avenues to conduct further theoretical analysis, as illustrated later in Theorem \ref{main Theorem}.
The following uses Example \ref{example guidance 2} to clarify how the NDI control technique contributes to constructing the incremental base policy and its accompanying benefits regarding the theoretical analysis.
\begin{example} \label{example guidance 2}
This example adopts the incremental base policy designed as 
\begin{equation} \label{eq exp base}
    \Delta u_{i_{b}} = \bar{g}^{\dagger}_i(\dot{x}_{d_{i}}-\dot{x}_{i,0}-k_i e_i),
\end{equation}
where $k_i  \in \mathbb{R}^{n_i \times n_i}$ is a predetermined constant matrix. 

The application of $\Delta u_{i_{b}}$ \eqref{eq exp base} on \eqref{eq err dyn} yields the incremental error dynamics
\begin{equation} \label{eq exp error}
        \dot{e}_i =  -k_i e_i + \bar{g}_{i} (\Delta u_{i_{r}} + \xi_{i}),
\end{equation}
that could be further stabilized by the incremental residual  policy $\Delta u_{i_{r}}$.
After that, focusing on \eqref{eq exp error}, the incremental residual policy $\Delta u_{i_{r}}$ would learn upon the incremental base policy \eqref{eq exp base} to further minimize the tracking error $e_i$.
This improves the tracking precision and the robustness of the incremental policy \eqref{incremental control}.
The explicit form of the incremental error dynamics presented in \eqref{eq exp error} offers us avenues to conduct theoretical analysis using the Lyapunov tool,  as illustrated later in Theorem \ref{main Theorem}.
\end{example}
Although the theoretical analysis on the basis of \eqref{eq exp error} is possible, practitioners should be aware of the expert knowledge and the engineering effort (debugging $k_i$) involved.
Besides, the fine-tuned traditional controllers for one specific task often lack generalization towards other different tasks. 

\subsubsection{Learning Perspective}
Alternatively, here the incremental base policy is implemented as a policy from imitation learning. 
In this case, the reasoning ability of human demonstrations is embedded into the incremental base policy that would guide the incremental residual policy to complete complex tasks.
In the following, Example \ref{example guidance 3} is provided to exemplify the incremental base policy constructed from expert demonstrations.

\begin{example} \label{example guidance 3}
This example uses behavioral cloning, one simple kind of imitation learning, to design the incremental base policy for an explanation.
For one certain task, assume that the demonstration dataset 
$ \mathcal{D} := \{ (s_j, a_j) | j = 1,2, \cdots \}$ is available, where $s_j$, $a_j$ are states and actions in suitable dimensions.
The base policy $u_{b} (\theta)$ parameterized by $\theta$ is learned by solving the following optimization problem:
\begin{equation} \label{eq exp BC}
  \theta^* = \arg  \min_{\theta} \sum_{(s_j,a_j) \sim \mathcal{D}}^{} \left\| a_j - u_{b}(\theta) \right\|^2.
\end{equation}
The required incremental base policy $\Delta  u_{i_{b}} $ for our work is gotten via the computation $\Delta  u_{i_{b}}  (k t_s) = u_{i_{b}}  \left( (k+1) t_s,  \theta^* \right) - u_{i_{b}} \left(k t_s,  \theta^* \right)$, where $k \in \mathbb{R}^+$.
\end{example}

The data-driven methods from the learning field are suitable for hard-to-specify tasks (challenging to be specified explicitly using rules or constraints) that are inefficient or even intractable for conventional controllers; especially in the case of available cheap (easily gotten) data. However, it is difficult to offer theoretical analysis given the non-interpretability of the utilized incremental base policy.

\begin{remark}
The incremental base policy designed from different techniques either in control or learning fields implies the modularity of our proposed tracking control scheme.
The property of the investigated problem, the available source (expert knowledge or data in particular), and the designers' preference (theoretical guarantee or task generalization) jointly determine the explicit method used to design the incremental base policy.
\end{remark}
Whether the incremental base policy is designed from control or learning perspectives, its adaptation ability towards dynamically changing environments is limited because the incremental base policy is often designed in static environments. 
The performance of the incremental base policy has already been determined by pre-collected data or prior-set controller parameters.
This motivates us to conduct further learning on the basis of the incremental base policy to improve performance in a dynamically changing environment.

\subsection{Incremental Residual Policy} \label{sec residual control}
This section utilizes the control-theoretical RL to develop the incremental residual policy that learns upon the incremental base policy to get improved task performance and robustness.
In particular, we use the incremental base policy in Section \ref{sec base policy} to initialize the incremental residual policy learning process.
Then, the incremental residual policy learns adaptions to the incremental base policy in an optimization process, where tracking errors and energy consumption are minimized for improved performance.
The incremental residual policy learned from the interactions with environments endows the resulting incremental policy \eqref{incremental control} with enhanced robustness towards dynamically changing environments.

For convenience, we represent the incremental error dynamics \eqref{eq err dyn} as 
\begin{equation} \label{eq error dyn 2 part new}
  \dot{e}_i =   \bar{f}_i + \bar{g}_{i} \Delta u_{i_{r}} +\bar{g}_i \xi_{i}, 
\end{equation}
where $\bar{f}_i := \dot{x}_{i,0} + \bar{g}_{i} \Delta u_{i_{b}}- \dot{x}_{d_{i}} \in \mathbb{R}^{n_i}$.
The explicit form of $\Delta u_{i_{b}}$ has been determined in Section \ref{sec base policy}.

In the following, we will learn $\Delta u_{i_{r}}$ 
to track the desired trajectory precisely (stabilizing the incremental error dynamics \eqref{eq error dyn 2 part new} interpreting from a control perspective).

Given $\xi_{i}$ in \eqref{eq error dyn 2 part new} is unknown, thus the available incremental error dynamics for later analysis follows
\begin{equation} \label{incrmental err system new}
    \dot{e}_{i} =\bar{f}_i + \bar{g}_i \Delta u_{i_{r}}.
\end{equation}

The following value function 
\begin{equation}\label{cost fuction}
V_{i}(t) = \int_{t}^{\infty} r_{i}(e_i(\nu),\Delta u_{i_{r}}(\nu))\,d\nu,
\end{equation}
is considered for the incremental residual policy learning process to enhance performance, 
where $r_i(e_i,\Delta u_{i_{r}}) := e^{\top}_i Q_i e_i + W_{i}(\Delta u_{i_{r}}) + \Bar{\xi}^2_{oi}$. 
The quadratic term $e^{\top}_i Q_i e_i$, where $Q_i \in \mathbb{R}^{n_i \times n_i}$ is a positive definite matrix, is introduced to improve the tracking precision.
The input penalty function $W_{i}(\Delta u_{i_{r}})$
follows
\begin{equation}\label{control penalty function}
    \mathcal{W}_i(\Delta u_{i_{r}}) =  2\int_{0}^{\Delta u_{i_{r}}}  \beta  \tanh^{-1}(\vartheta / \beta)  \,d\vartheta,
\end{equation}
which is utilized to punish and enforce the incremental residual policy as $\left\|\Delta u_{i_{r}}\right\| \leq \beta \in \mathbb{R}^+$.
The limited $\Delta u_{i_{r}}$ is beneficial since a severe interruption might lead to an abrupt change 
of $\Delta u_{i_{r}}$, which might destabilize the learning process introduced in Section \ref{sec Approximate solutions}.
The estimation error related term follows $\Bar{\xi}_{oi} = \bar{c}_i  \left\|\Delta u_{i_{r}}\right\|$, where $\bar{c}_i \in \mathbb{R}^+$ is chosen as illustrated in Theorem \ref{theorem equalivance}. 
Note that $\Bar{\xi}_{oi}$ is introduced to account for the influence of the estimation error $\xi_{i}$ (temporally ignored in \eqref{incrmental err system new}) on the learning process.
The proof given in Theorem \ref{theorem equalivance} illustrates the rationality of incorporating $\Bar{\xi}_{oi}$ into the value function to address the estimation error during the optimization process.

For $\Delta u_{i_{r}} \in \Psi$, where $\Psi$ is the set of admissible incremental control policies \cite[Definition 1]{li2021model}, the associated optimal value function follows
\begin{equation} \label{optimal value function}
    V^{*}_i = \min_{\Delta u_{i_{r}} \in \Psi}\int_{t}^{\infty}  r_i(e_i(\nu),\Delta u_{i_{r}}(\nu))\,d\nu.
\end{equation}
Define the Hamiltonian function as 
\begin{equation} \label{Hamiltonian function}
    \small{H_i(e_i,\Delta u_{i_{r}},\nabla V_i) = r(e_i,\Delta u_{i_{r}}) + \nabla V^{{T}}_i( \bar{f}_i + \bar{g}_{i} \Delta u_{i_{r}})},
\end{equation}
where $\nabla (\cdot) := \partial (\cdot) / \partial e_i$. Then, $V^{*}_i$ satisfies the Hamilton-Jacobi-Bellman(HJB) equation
\begin{equation} \label{HJB equation}
    0 = \min_{\Delta u_{i_{r}} \in \Psi} [H_i(e_i,\Delta u_{i_{r}},\nabla V^{*}_i)].
\end{equation}
Assume that the minimum of \eqref{optimal value function} exists and is unique \cite{li2021model,vamvoudakis2010online}. By using the stationary optimality condition on the HJB equation \eqref{HJB equation}, we get the optimal incremental residual policy
\begin{equation} \label{optimal incremental u}
    \Delta u^{*}_{i_{r}} = - \beta \tanh(\frac{1}{2\beta} \bar{g}^{\top}_i \nabla V^{*}_i),
\end{equation}
in an analytical form.
To obtain $\Delta u^{*}_{i_{r}}$, we need to solve the HJB equation \eqref{HJB equation} to determine the value of $\nabla V^{*}_i$, which is detailly clarified in Section \ref{sec Approximate solutions}.

In the following part of this subsection, based on the estimation error bound given in Lemma \ref{boud of TDE AE}, we prove in Theorem \ref{theorem equalivance} that 
the incremental residual policy $\Delta u^{*}_{i_{r}}$ \eqref{optimal incremental u} for \eqref{incrmental err system new} robustly stabilize the incremental error dynamics \eqref{eq error dyn 2 part new}.

\begin{lemma} \label{boud of TDE AE}
Given a sufficiently high sampling rate, $\exists \Bar{\xi}_i \in \mathbb{R}^+$, there holds $\left\|\xi_i\right\| \leq \Bar{\xi}_i$. 
\end{lemma}
\begin{proof}
The proof is available in Appendix~\ref{app error bound}. 
\end{proof}
\begin{theorem}\label{theorem equalivance}
Consider the incremental error dynamics \eqref{eq error dyn 2 part new} with a sufficiently high sampling rate, 
 if there exists a scalar $\bar{c}_i \in \mathbb{R}^{+}$ such that the following inequality is satisfied
     \begin{equation} \label{equalivance condition}
    \bar{\xi}_i < \bar{c}_i  \left\| \Delta u_{i_{r}} \right\|,
    \end{equation}
  the optimal incremental residual policy \eqref{optimal incremental u} regulates the tracking error to a small neighborhood around zero while minimizing the value function \eqref{cost fuction}.
\end{theorem}
\begin{proof}
See Appendix \ref{app equalivance proof}. 
\end{proof}
Theorem \ref{theorem equalivance} implies that the incremental residual policy $\Delta u^{*}_{i_{r}}$ \eqref{optimal incremental u} robustly stabilize \eqref{eq error dyn 2 part new} in an optimal manner.
Thus, the combination with the incremental residual policy $\Delta u^{*}_{i_{r}}$  and the incremental base policy $\Delta u_{i_{b}}$ solves Problem \ref{Problem 2} together.

\begin{remark}
    The robot working at a sufficiently high sampling rate is the prerequisite of Lemma \ref{boud of TDE AE}.  
    The high sampling rate is only possible for methods with low computational complexity. 
     To make our proposed DR-RL method work at a high sampling rate,
     this work uses a parallel learning architecture to distribute the intensive computation load into multiple processors;
     and seeks an analytical solution to the incremental residual policy that favors low computation complexity.
\end{remark}

\section{Parallel Computation of Approximate Solution} \label{sec Approximate solutions}
This section presents parallelized critic agents for learning the approximate solution to $\nabla V^{*}_i$ in the HJB equation \eqref{HJB equation} via a computationally efficient parallel way.
Thereby, an approximation to the optimal incremental residual policy \eqref{optimal incremental u} is obtained.
The exploitation of realtime and experience data together facilitates one simple yet efficient off-policy critic neural network (NN) weight update law with guaranteed weight convergence and improved sample efficiency.
\subsection{Value Function Approximation} \label{sec vaule function appr}
For $e_i \in \Omega$, where $\Omega \subset \mathbb{R}^{n_i}$ is a compact set, the $i$th continuous optimal value function \eqref{optimal value function} is approximated by $i$th parallelized critic agent as \cite{vamvoudakis2010online}
\begin{equation}\label{optimal V approximation}
    V^{*}_{i} ={W_{i}^*}^{\top} \Phi_i(e_i) + \epsilon_i(e_i),
\end{equation}
where ${W_{i}^*} \in \mathbb{R}^{N_i}$, $\Phi_i(e_i): \mathbb{R}^{n_i} \to \mathbb{R}^{N_i}$, and $\epsilon_i(e_i) \in \mathbb{R}$ denote the NN weight, the activation function, and the approximation error of the $i$th parallelized critic agent, respectively.
\begin{remark}
The utilized decoupling technique in Section \ref{sec incremental subsystem} solves the curse of complexity problem in value function approximation. 
In particular, the size of the constructed $i$th critic NN \eqref{optimal V approximation} relies on the dimension of the local error $e_i$. 
The $n_i$-D $e_i$ allows us to construct a low-dimensional $\Phi_i(e_i)$ (easy to choose) 
\footnote{It is displayed in Sections \ref{sec simulation} and \ref{section Experimental validation} that 4-D activation functions $\Phi_i(e_i)$ in a fixed structure are chosen for subsystems of 2-DoF, 3-DoF, and 7-DoF robot manipulator, and also a 6-DoF quadrotor in the appendix.} 
to approximate its associated $V^{*}_{i}$ regardless of the value of the original robot dimension $n$. Otherwise, for a global approximation, i.e., $V^{*} ={W^*}^{\top} \Phi(e) + \epsilon(e)$ with  $e := x-x_d \in \mathbb{R}^{n}$, the dimension of $\Phi(e)$ increases exponentially as $n$ increases.
\end{remark}

To facilitate the later theoretical analysis, here we provide an assumption that is common in related works \cite{vamvoudakis2010online}.
 \begin{assumption} \label{bound of NN issues}
    There exist constants $b_{\epsilon_{i}}, b_{\epsilon_{ei}}, b_{\epsilon_{hi}}, b_{\Phi_i}, b_{\Phi _{ei}}\in \mathbb{R}^{+}$ such that $\left\| \epsilon_i(e_i)  \right\| \leq b_{\epsilon_{i}}$, $\left\| \nabla\epsilon_i(e_i)  \right\| \leq b_{\epsilon_{ei}}$, $\left\| \epsilon_{hi} \right\| \leq b_{\epsilon_{hi}}$, $\left\| \Phi_i(e_i) \right\| \leq b_{\Phi_i}$, and $\left\| \nabla\Phi_i(e_i) \right\| \leq b_{\Phi _{ei}}$.
 \end{assumption} 
     
    Given a fixed $i$th incremental residual policy $\Delta u_{i_{r}}$, combining \eqref{HJB equation} with  \eqref{optimal V approximation}  yields
    \begin{equation}\label{approximation Lyapunov equation}
      {W^*_{i}}^{\top}\nabla \Phi_{i}( \bar{f}_i + \bar{g}_i \Delta u_{i_{r}} )+r_i(e_i,\Delta u_{i_{r}}) = \epsilon_{h_{i}},
    \end{equation}
    where the $i$th residual error follows $\epsilon_{h_{i}} := -\nabla \epsilon^{\top}_i(  \bar{f}_i + \bar{g}_i \Delta u_{i_{r}})  \in \mathbb{R}$.
    The NN parameterized \eqref{approximation Lyapunov equation} is 
    rewritten as
    \begin{equation}\label{LIP Lyapunov equation}  
    \Theta_i = -{W^*_i}^{\top}Y_i+\epsilon_{h_i},
    \end{equation}
    where $\Theta_i := r_i(e_i,\Delta u_{i_{r}}) \in \mathbb{R}$, and $Y_i := \nabla \Phi_i(\bar{f}_i + \bar{g}_i \Delta u_{i_{r}}) \in \mathbb{R}^{N_i}$. This formulated linear in parameter form 
    simplifies the development of an efficient NN weight update law in the subsequent subsection.
    \subsection{Off-Policy Critic NN Weight Update Law} \label{sec Off-policy weight update law}
    An approximation of \eqref{LIP Lyapunov equation} follows
    \begin{equation}\label{approximation LIP Lyapunov equation}
    \hat{\Theta}_i = -\hat{W}_i^{\top}Y_i,
    \end{equation}
    where $\hat{W}_i \in \mathbb{R}^{N_i}$, $\hat{\Theta}_i \in \mathbb{R}$ are estimates  of $W_i^*$ and $\Theta_i$, respectively.
    To achieve $\hat{W}_i \to W_i^{*}$, we design the off-policy critic NN weight update law
    \begin{equation} \label{w update law}
        \dot{\hat{W}}_i = - \Gamma_i k_{t_{i}} Y_i\Tilde{\Theta}_i -  \sum_{l=1}^{P_{i}} \Gamma_{i} k_{e_{i}} Y_{{i_{l}}}\Tilde{\Theta}_{{i_{l}}},
    \end{equation}
    for the $i$th parallelized critic agent to learn the NN weight $\hat{W}_i$ in a parallel way via minimizing $E_i := \frac{1}{2} \Tilde{\Theta}_i^{\top}\Tilde{\Theta}_i$, where $\Tilde{\Theta}_i := \Theta_i - \hat{\Theta}_i\in \mathbb{R}$. Here $\Gamma_{i} \in \mathbb{R}^{N_i \times N_i}$ is a constant positive definite gain matrix; $k_{t_{i}}, k_{e_{i}} \in \mathbb{R}^{+}$ are used to trade-off the contribution of realtime and experience data to the online NN weight learning process; $P_i \in \mathbb{R}^{+}$ is the number of recorded experience data. 

    To guarantee the weight convergence of \eqref{w update law}, as proved in Lemma \ref{Theorem weight convergence}, the exploited experience data should be sufficiently rich to satisfy the rank condition in Assumption~\ref{rank condition}, which could be easily satisfied by sequentially reusing experience data \cite{li2021model}.
    \begin{assumption} \label{rank condition}
          Given an experience buffer $\mathfrak{B}_{i} = [Y_{i_{1}},...,Y_{i_{P_{i}}}] \in \mathbb{R}^{N_i \times P_i}$, 
          there holds $rank(\mathfrak{B}_{i}) = N_i$.
    \end{assumption} 
\begin{lemma}\label{Theorem weight convergence}
Given Assumption~\ref{rank condition}, the NN weight learning error $\Tilde{W}$ converges to a small neighborhood around zero.
\end{lemma}
\begin{proof}
The proof is similar to our previous work \cite{li2021model}, which relates with the optimal regulation control problem. Thus, it is omitted here due to page limits.
\end{proof}
The guaranteed weight convergence of $\hat{W}_i$ to $W_i^*$ presented in Lemma \ref{Theorem weight convergence} permits us to adopt a computation-simple single critic NN structure, where the estimated critic NN weight $\hat{W}_i$ is directly used to construct the approximate optimal incremental residual policy:
\begin{equation} \label{optimal incremental hat u}
    \Delta \hat{u}_{i_{r}}= - \beta \tanh(\frac{1}{2\beta}\bar{g}^{\top}_i \nabla \Phi_{i}^{\top}\hat{W}_i). 
\end{equation}
 \begin{figure*}[!t]
    \centering
    \includegraphics[width=7.2in]{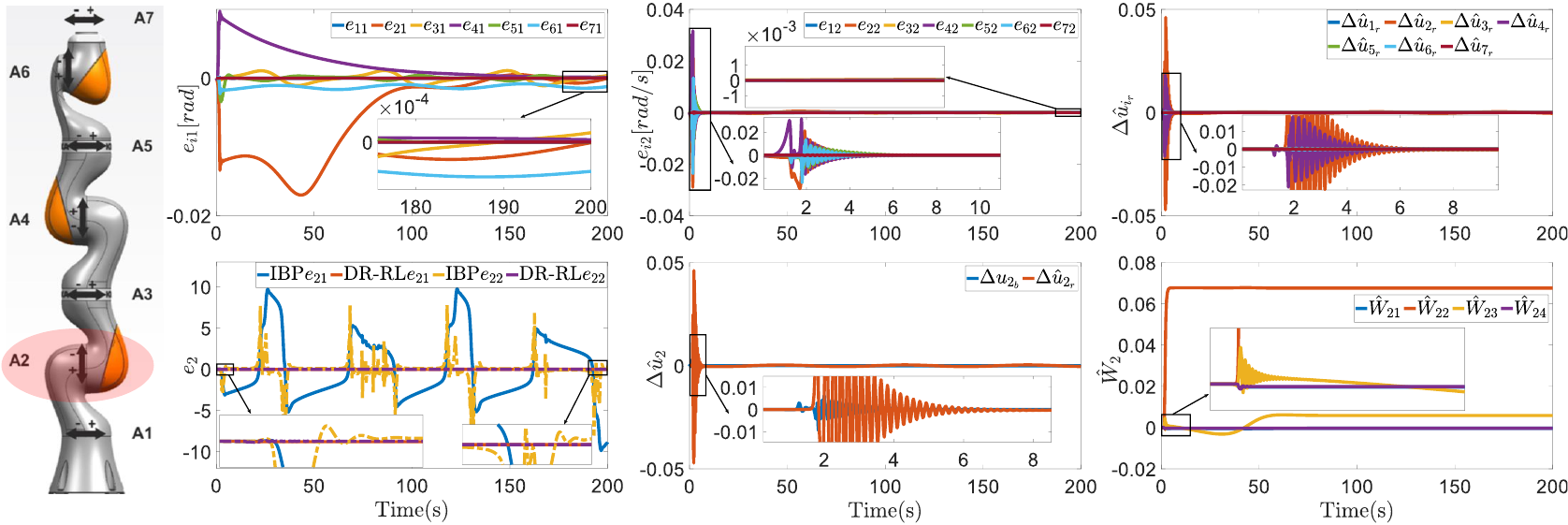}
    \caption{
    The numerical validation results on the 7-DoF KUKA iiwa robot manipulator.
    Top: the evolution trajectories of the $i$th subsystem's tracking error $e_{i_{1}}$, $e_{i_{2}}$ and the learned IRP $\Delta \hat{u}_{i_{r}}$, $i = 1,2,\cdots,7$;
    Bottom: the comparative simulation results focusing on the $2$nd subsystem, including the associated error trajectories of both incremental base policy (IBP) and DR-RL cases, the evolution trajectories of the IBP $\Delta \hat{u}_{i_{r}}$ and the incremental residual policy $\Delta \hat{u}_{i_{r}}$, and the weight convergence result.} 
    \label{kuka traj of w and e}
\end{figure*}
Finally, combining with \eqref{incremental control} and \eqref{optimal incremental hat u} gets the robotic policy
\begin{equation} \label{optimal  overall u}
     \hat{u}_i =  u_{i,0} +  \Delta u_{i_{b}}+ \Delta \hat{u}_{i_{r}},
\end{equation}
applied at the $i$th subsystem \eqref{i-th state space system}.
Based on the theoretical analysis mentioned above, we provide the main conclusions in Theorem \ref{main Theorem}.
\begin{figure*}[!t]
    \centering
    \includegraphics[width=7.2in]{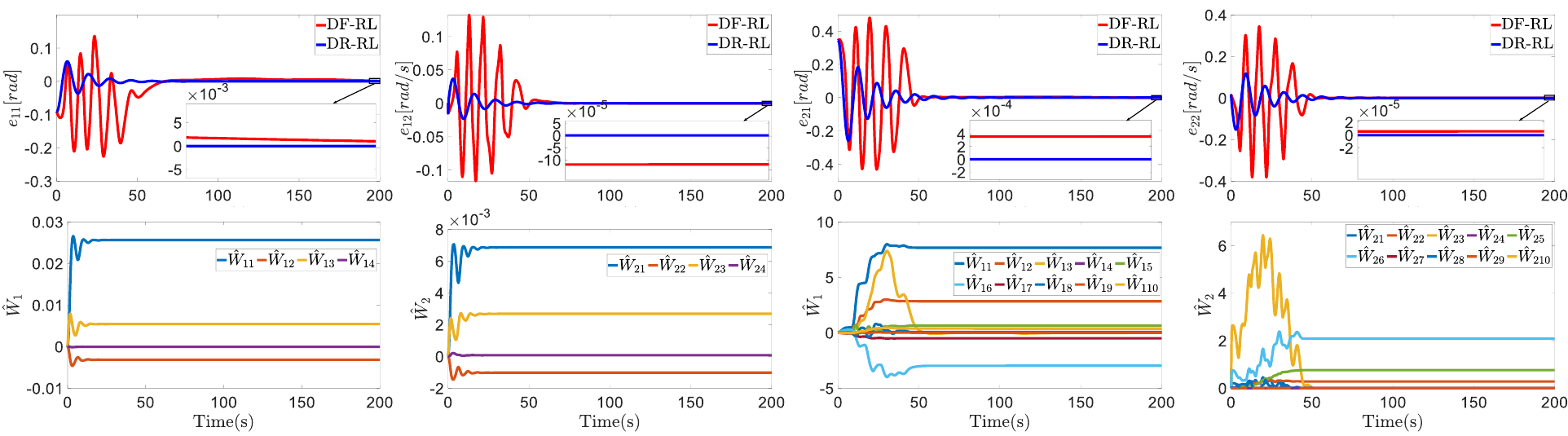}
    \caption{
    The comparative numerical simulation results about the task-space task of the 2-DoF robot manipulator.
    Top: the evolution trajectories of the $i$th subsystem's tracking errors $e_{i_{1}}$, $e_{i_{2}}$ based on the DF-RL and DR-RL methods, $i = 1,2$;
    Bottom: the evolution trajectories of the $i$th subsystem's 4-D estimated NN weight $\hat{W}_i = [\hat{W}_{i_{1}},\cdots,\hat{W}_{i_{4}}]^{\top}$ for the  DR-RL method, and the $i$th subsystem's
    10-D estimated NN weight $\hat{W}_i = [\hat{W}_{i_{1}},\cdots,\hat{W}_{i_{10}}]^{\top}$  for the DF-RL method, $i = 1,2$.} 
    \label{traj of w and e 2dof task space}
\end{figure*}
 \begin{theorem}\label{main Theorem}
Given Assumptions \ref{bound of NN issues}--\ref{rank condition}, for a sufficiently large $N_i$, 
the off-policy critic NN weight update law \eqref{w update law},
and the approximate optimal incremental residual policy \eqref{optimal incremental hat u} guarantee the tracking error and the NN weight learning error uniformly ultimately bounded (UUB).
\end{theorem}
\begin{proof}
See Appendix \ref{app stability proof}.
\end{proof}
\section{Comparative Numerical Simulation} \label{sec simulation}
This section presents comparative numerical validations to show the superiority of our proposed DR-RL method over other methods.
We firstly compare the performance of the incremental base policy and the DR-RL based control strategies on the basis of the high-dimensional 7-DoF KUKA iiwa robot manipulator and 6-DoF quadrotor, which clearly present the improved task performance induced by the online learning part.
Then, we compare with two fine-tuned baselines \cite{modares2014optimal,kamalapurkar2015approximate} focusing on a 2-DoF benchmark robot manipulator to show the superiority of our proposed DR-RL over the baseline methods regarding sample efficiency, tracking performance, and task flexibility.

\subsection{Validations on high-dimensional robotics} 
This subsection fully validates the performance of our proposed DR-RL method based on the high-dimensional robotic systems: the 7-DoF KUKA iiwa robot manipulator in Matlab Robotics System Toolbox
and the commonly used benchmark 6-DoF quadrotor platform. 

A desired trajectory $x_d = [q_d^{\top}, \dot{q}^{\top}_d]^{\top}  \in \mathbb{R}^{14}$, wherein $q_d =(1+\sin{(\frac{t}{10}-\frac{\pi}{2}}))k_{p_{1}} \in \mathbb{R}^7$ with $k_{p_{1}} = 0.1\boldsymbol{I}_{7}$ is adopted to simulate a manipulation task for the KUKA iiwa robot manipulator. We adopt the 4-D activation function $\Phi_i(e_i)=[e^{2}_{i_{1}},e^{2}_{i_{2}},e_{i_{1}}e_{i_{2}},e^{3}_{i_{2}}]^{\top}$  for the accurate value function approximation of the $i$th subsystem, $i = 1,2,\cdots,7$.
The detailed parameter settings of the incremental base policy and the DR-RL based control policies are referred to in Table I in the appendix.

The error trajectories of the joint angles and angle velocities present in Fig.~\ref{kuka traj of w and e} fully validate the performance of our proposed method. Furthermore, the comparative results focusing the $2$th subsystem 
(difficulty in control due to gravity) 
exemplify the enhanced task performance from the online learning part, which realizes the desired weight convergence as presented in Fig.~\ref{kuka traj of w and e}.

Furthermore, we test the generality of our proposed DR-RL method on different high-dimensional robotic systems via the trajectory tracking task of one 6-DoF benchmark quadrotor platform.
The detailed parameter settings and results are referred to in  Appendix~\ref{sim 6dof uav}, which are omitted here due to the page limit.

\subsection{Validations of Sample Efficiency and Task Flexibility} 
This subsection considers a task-space circle tracking task (centered at $c = (1,1)$ with radius $r = 0.5$) of a 2-DoF robot manipulator (the detailed model information is referred to in Appendix~\ref{sim setting}) for our proposed DR-RL method and one learning-from-scratch RL approach that uses a discounted factor suppressed value function \cite{modares2014optimal} to learn the approximate optimal tracking control policy (referred to as DF-RL for simplicity).
The associated joint-space
reference trajectory $x_d = \left[ q^{\top}_d, \dot{q}^{\top}_d \right]^{\top} \in \mathbb{R}^4$ is calculated through analytical inverse kinematics. 
The DF-RL \footnote{We introduce the parallel learning architecture to reformulate the method proposed in \cite{modares2014optimal} to make it work on a 2-DoF robot manipulator.} approach adopts a 10-D activation function for the accurate value function approximation of each subsystem, which is referred to in Appendix~\ref{sim setting}.
While our proposed  DR-RL approach only requires a 4-D activation function 
$\Phi_i(e_i)=[e^{2}_{i_{1}},e^{2}_{i_{2}},e_{i_{1}}e_{i_{2}},e^{3}_{i_{2}}]^{\top}$ 
for each subsystem, same as the ones used for KUKA iiwa robot manipulator and the quadrotor.
The sampling rate is chosen as $1\si{\kHz}$. 
The parameter settings for the DR-RL and DF-RL approaches are referred to in Table III in the Appendix~\ref{sim setting}.

A broader spatial variance during the initial learning period is observed for the DF-RL approach in the top four subfigures in Fig.~\ref{traj of w and e 2dof task space}. This is undesirable for hardware deployments.
While our developed DR-RL approach explores a smaller set of states for the online learning process and converges faster.
This is because the guidance from the incremental base policy helps decrease the exploration space.
Besides, we observe that our proposed DR-RL approach realizes higher tracking accuracy than the DF-RL approach which learns from scratch.
This benefits from the residual learning mechanism.
We observe from the bottom four subfigures in Fig.~\ref{traj of w and e 2dof task space} that the weight trajectories of our developed DR-RL approach converge faster than the DF-RL approach. The faster convergence rate implies  less required  data  for learning one satisfying policy. This validates the improved sample efficiency brought by the residual formulation and the MBRL mechanism utilized in our work.

To validate the enhanced task flexibility of our
proposed DR-RL method over the baseline RL-based tracking control strategy \cite{kamalapurkar2015approximate} developed under an augmented system, we  focuses on a complex robotic manipulation task composed of trajectories represented as different trajectory dynamics. The detailed results are referred to in Appendix~\ref{sim dof 2}, which is omitted due to page limits.

    \begin{figure*}[!t]
    \centering
    \includegraphics[width=7.2in]{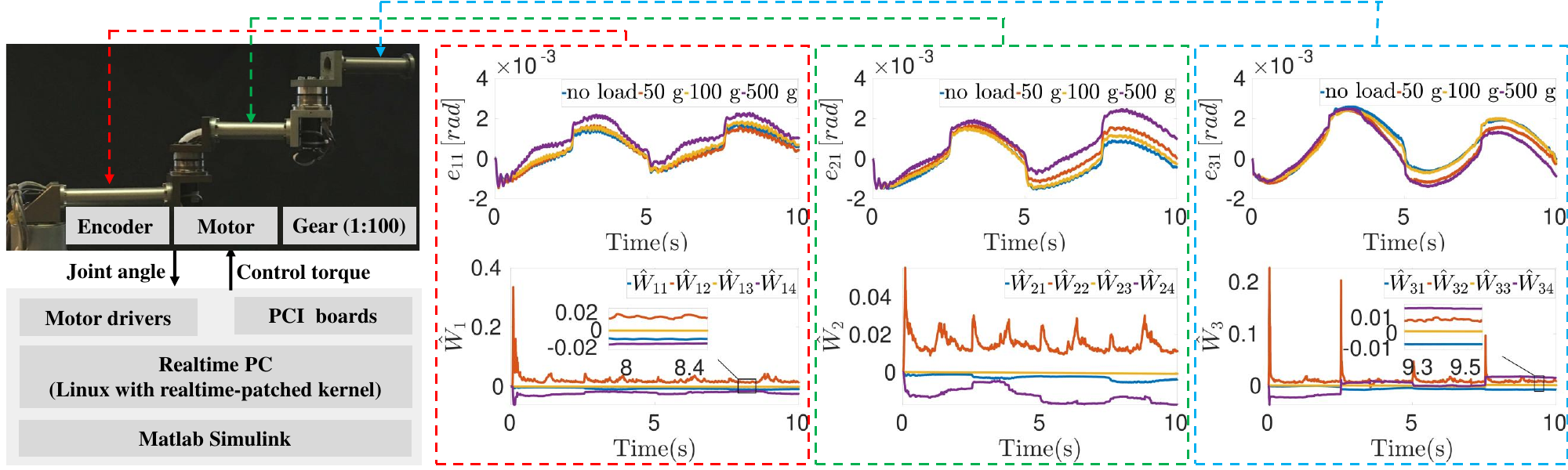}
    \caption{
    The experimental validation results about the 3-DoF robot manipulator.
    Top: the evolution trajectories of the $i$th subsystem's tracking error $e_{i_{1}}$ under different payloads, $i = 1,2,3$;
    Bottom: the evolution trajectories of the $i$th subsystem's 4-D estimated weight $\hat{W}_i = [\hat{W}_{i1},\cdots,\hat{W}_{i4}]^{\top}$, $i = 1,2,3$ for the $500~\si{\gram}$ payload case.
    }  
    \label{traj of w and e}
\end{figure*}

\section{Experimental validation} \label{section Experimental validation}
This section experimentally validates the efficiency of our proposed DR-RL based tracking control policy on a 3-DoF robot manipulator (see Fig.~\ref{traj of w and e}). 
The detailed hardware information is provided in our previous work \cite{li2021concurrent}.

During the experiment, the robot manipulator is driven to track the desired piecewise trajectory $x_d = [q_d^{\top}, \dot{q}^{\top}_d]^{\top}  \in \mathbb{R}^{6}$ with
    $q_d = (1+\sin{(\frac{t}{2}-\frac{\pi}{2}}))k_{p_{2}} \in \mathbb{R}^3$,
where $k_{p_{2}} = [0.3,0.6,1]^{\top}$ for $t \in \left [0, 5\right )$, and $k_{p_{2}}  = [0.2,0.5,0.8]^{\top}$ for $t \in \left [5, 10\right]$.
The sampling rate is set as $1\si{\kHz}$. 
Note that neither DF-RL nor AS-RL based tracking control policy is intractable to complete the tracking task provided here.
This is because it is nontrivial to find
the high-dimensional activation function required for accurate value function approximation of DF-RL and AS-RL methods.
Even though a  high-dimensional activation function is available, the realtime performance of the corresponding weight update law is poor for practical experiments.

Regarding the incremental residual policy, 
we choose the 4-D activation function $\Phi_i(e_i)=[e^{2}_{i_{1}},e^{2}_{i_{2}},e_{i_{1}}e_{i_{2}},e^{3}_{i_{2}}]^{\top}$ (same as the one for the 7-DoF KUKA iiwa robot manipulator) for the value function approximation of the $i$th subsystem, $i = 1,2,3$. The utilized low-dimensional activation function $\Phi_i(e_i)$ in a fixed structure exemplifies our method's scalability and practicability towards different robotic systems.
The parameters for subsystems 1-3 are set as:
$Q_i = \diag{[300,40000]}$, $\bar{c}_i = 200$,  
$\Gamma_i =  \diag{[100,4,0.1,16]}$, $k_{t_{i}} = 0.2$, $k_{e_{i}} = 0.01$, $P_i = 10$,
$k_i = \diag{[8,8]}$, $i = 1,2,3$;
and $\beta = 0.1$, $\bar{g}_1 =40$,  $\bar{g}_2 =46$, and $\bar{g}_3 =54$.

The trajectories of $e_{i_{1}}$,  $i = 1,2,3$ under different payloads (installed to the end effector of the robot manipulator) are displayed in the top three subfigures in Fig.~\ref{traj of w and e}. It is shown that our developed tracking control scheme efficiently tracks the desired trajectories with satisfying tracking precision and robustness against varying payloads. 
Three subsystems' $\hat{W}_i$ of the $500~\si{\gram}$ payload case, which are trained in parallel using realtime and experience data together, are displayed in the bottom three subfigures in Fig.~\ref{traj of w and e}. 
We obtain the desired weight convergence for each subsystem. This validates the realtime performance of our developed weight update law
\eqref{w update law} and the efficiency of the parallel learning mechanism.

\section{Conclusion} \label{section conclusion}
This work develops a sample-efficient and scalable DR-RL method applicable to high-dimensional robotic tracking control tasks.
The sample efficiency is improved by guiding the learning process with expert knowledge or human demonstration, and using both implicit model information and off-policy experience data for the learning process. 
The scalability towards high-dimensional robots is through the parallel learning architecture, wherein critic agents learn in parallel to jointly solve the task.
Comparative numerical and experimental results validate the effectiveness of our proposed DR-RL approach.
In the future, the proposed method remains to  be further improved from the following perspectives: addressing the input saturation problem, extending to control nonaffine systems, investigating estimation methods for indirectly measurable states, and online learning of the gain matrix $\bar{g}_i$ in \eqref{incrmental system}.
\bibliographystyle{IEEEtran}
\bibliography{bibtex/bib/main}

\appendices

\section{Proof of Theorem 1} \label{app equalivance proof}
\begin{proof}
Considering that $V^{*}_i(e_i=0) =0$,  and $V^{*}_i >0$ for $\forall e_i \ne 0$, $V^{*}_i$ in \eqref{optimal value function} could serve as a candidate Lyapunov function. Taking time derivative of $V^{*}_i$ along the $i$th incremental error subsystem \eqref{eq error dyn 2 part new} yields
 \begin{equation} \label{equalivance eq1}
    \begin{aligned}
    \dot{V}^{*}_i = \nabla V^{*}_i ( \bar{f}_i + \bar{g}_i \Delta u^*_{i_{r}}) + \nabla V^{*}_i \bar{g}_i \xi_{i}. 
    \end{aligned}
\end{equation}
According to  \eqref{Hamiltonian function} and \eqref{HJB equation}, the following equations establish
 \begin{equation} \label{equalivance eq2}
    \begin{aligned}
    \nabla V^{*}_i (\bar{f}_i + \bar{g}_i \Delta u^*_{i_{r}}) &=  -e^{\top}_i Q_i e_i - W_{i}(\Delta u^*_{i_{r}}) - \Bar{\xi}^2_{oi}\\
    \nabla V^{*}_i \bar{g}_i &= -2 \beta \tanh^{-1}(\Delta u^*_{i_{r}}/\beta).\\ 
    \end{aligned}
\end{equation}
Substituting \eqref{equalivance eq2} into \eqref{equalivance eq1} yields
 \begin{equation} \label{equalivance eq3}
    \begin{aligned}
    \dot{V}^{*}_i = -e^{\top}_i Q_i e_i - W_{i}(\Delta u^*_{i_{r}}) - \Bar{\xi}^2_{oi}   -2 \beta \tanh^{-1}(\Delta u^*_{i_{r}}/\beta)\xi_{i}. 
    \end{aligned}
\end{equation}
As for the $W_{i}(\Delta u^*_{i_{r}})$ in \eqref{equalivance eq3}, according to our previous result \cite[Theorem 1]{li2021model}, it follows that 
 \begin{equation} \label{equalivance eq4}
    \begin{aligned}
    W_{i}(\Delta u^*_{i_{r}}) = \beta^2 \sum_{j=1}^{m} \left(\tanh^{-1}(\Delta u^*_{i_{r}}/\beta)\right)^2 - \epsilon_{u{_i}},
    \end{aligned}
\end{equation}
where  $\epsilon_{u{_i}}  \leq  \frac{1}{2} \bar{g}^2_i \nabla {V_i^*}^{\top}\nabla V_i^{*}$. 
Given that there exists $b_{\nabla V_i^*} \in \mathbb{R}^+$ such that $\left\|\nabla V_i^{*}\right\| \leq b_{\nabla V_i^*}$. Thus, we could rewrite the bound of $\epsilon_{u{_i}}$ as $\epsilon_{u{_i}}  \leq b_{\epsilon_{ui}}  \leq \frac{1}{2} \bar{g}_i^2 b^2_{\nabla V_i^*}$.

Then, substituting \eqref{equalivance eq4} into \eqref{equalivance eq3}, we get
 \begin{equation} \label{equalivance eq5}
    \begin{aligned}
    \dot{V}^{*}_i = &-e^{\top}_i Q_i e_i  -[\beta \tanh^{-1}(\Delta u^*_{i_{r}}/\beta)+\xi_i]^2 \\
   & - (\Bar{\xi}^2_{oi}-\xi_i^{\top}\xi_i ) + b_{\epsilon_{ui}}.
    \end{aligned}
\end{equation}
We choose $\Bar{\xi}_{oi} =  \bar{c}_i \left\| \Delta u_{i_{r}} \right\|$, and $\bar{c}_i$ is picked to satisfy $\bar{c}_i \left\| \Delta u_{i_{r}} \right\| > \Bar{\xi}_i$, where $\Bar{\xi}_i$ is defined in Lemma 1. Then, the following equation holds
 \begin{equation} \label{equalivance eq6}
    \begin{aligned}
    \dot{V}^{*}_i \leq -e^{\top}_i Q_i e_i +  b_{\epsilon_{ui}}.
    \end{aligned}
\end{equation}
Thus, if $-\lambda_{\min}(Q_i) \left\| e_i \right\|^2 +  b_{\epsilon_{ui}} <0$, $\dot{V}^{*}_i <0$ holds.
Here $\lambda_{\min}(\cdot)$ denotes the minimum eigenvalues of a symmetric real matrix.
Finally, it concludes that states of the $i$th incremental error subsystem \eqref{eq error dyn 2 part new} converges to the residual set
\begin{equation} \label{equalivance eq9}
\Omega_{e_i} = \{e_i | \left\| e_i \right\| \leq \sqrt{b_{\epsilon_{ui}}/\lambda_{\min}(Q_i)} \}.
\end{equation}
This concludes the proof.
\end{proof}

\section{Proof of Theorem 3} \label{app stability proof}

\begin{proof}
Consider the candidate Lyapunov function for the $i$th incremental error subsystem \eqref{eq error dyn 2 part new} as 
    \begin{equation} \label{Lya function stability}
    L_i =  V^{*}_{i} + \frac{1}{2} \Tilde{W}^{\top}_{i} \Gamma_i^{-1} \Tilde{W}_{i}.
    \end{equation}
By denoting $L_{i_{1}} =V^{*}_{i}$, \emph{its derivative follows}
\begin{equation} \label{stability dLv 1}
\begin{aligned}
\dot{L}_{i_{1}} &=\nabla {V_i^*}^{\top}( \bar{f}_i + \bar{g}_i \Delta \hat{u}_{i_{r}} + \bar{g}_i \xi_{i})\\
                 & = \nabla {V_i^*}^{\top}( \bar{f}_i + \bar{g}_i \Delta u^*_{i_{r}}) +\nabla {V_i^*}^{\top} \bar{g}_i\xi_i \\
                 & \; \;\;\;+ \nabla {V_i^*}^{\top}\bar{g}_i(\Delta \hat{u}_{i_{r}}-\Delta u^*_{i_{r}}).
\end{aligned}
\end{equation}
Substituting \eqref{equalivance eq2} into \eqref{stability dLv 1} reads
\begin{equation} \label{stability dLv 2}
\begin{aligned}
\dot{L}_{i_{1}} 
     & = -e^{\top}_i Q_i e_i - \mathcal{W}(\Delta u^{*}_{i_{r}}) - \Bar{\xi}^2_{oi} - 2 \beta \tanh^{-1}\left(\Delta u^{*}_{i_{r}}/\beta\right) \xi_i \\
     & \; \;\;\; -2 \beta \tanh^{-1}\left(\Delta u^{*}_{i_{r}}/\beta\right)( \Delta \hat{u}_{i_{r}}-\Delta u^*_{i_{r}}).
\end{aligned}
\end{equation}
Combining with \eqref{equalivance eq4} and \eqref{equalivance eq5}, \eqref{stability dLv 2} follows
\small
\begin{equation} \label{stability dLv 3}
\begin{aligned}
    \dot{L}_{i_{1}} 
     \leq &-e^{\top}_i Q_i e_i - (\Bar{\xi}^2_{oi}-\left\|\xi_i\right\|^2 ) - \left[\beta \tanh^{-1}\left(\Delta u^{*}_{i_{r}}/\beta\right)+\xi_i\right]^2 \\
     & +\frac{1}{2} \nabla {V_i^*}^{\top}\bar{g}_i \bar{g}^{\top}_i\nabla V^{*}_{i} -2 \beta \tanh^{-1}\left(\Delta u^{*}_{i_{r}}/\beta\right)( \Delta \hat{u}_{i_{r}}-\Delta u^*_{i_{r}}).
\end{aligned}
\end{equation}
\normalsize

The term $-2 \beta \tanh^{-1}\left(\Delta u^{*}_{i_{r}}/\beta\right)( \Delta \hat{u}_{i_{r}}-\Delta u^*_{i_{r}})$ in \eqref{stability dLv 3} follows 
\small
\begin{equation} \label{stability dLv 4}
\begin{aligned}
-2 \beta \tanh^{-1}\left(\Delta u^{*}_{i_{r}}/\beta\right)( \Delta \hat{u}_{i_{r}}-\Delta u^*_{i_{r}})
\leq &\beta^2 \left\| \tanh^{-1}\left(\Delta u^{*}_{i_{r}}/\beta\right) \right\|^2  \\
&+ \left\| \Delta \hat{u}_{i_{r}}-\Delta u^*_{i_{r}} \right\|^2.
\end{aligned}
\end{equation}
\normalsize
According to \eqref{optimal incremental u} and \eqref{optimal V approximation}, and the mean-value theorem, the optimal incremental control is rewritten as
\begin{equation} \label{stability dLv 41}
\Delta u^*_{i_{r}} = -\beta\tanh(\frac{1}{2\beta}\bar{g}^{\top}_i\nabla \Phi^{\top}_{i}W^*_{i})-\epsilon_{\Delta u^*_{i}},
\end{equation}
where $\epsilon_{\Delta u^*_{i}} = \frac{1}{2} (1-\tanh^2(\eta_i))\bar{g}^{\top}_i\nabla \epsilon_i$, and $\eta_i \in \mathbb{R}$ is chosen between $\frac{1}{2\beta}\bar{g}^{\top}_i\nabla \Phi^{\top}_{i}W^*_{i}$ and $\frac{1}{2\beta}\bar{g}^{\top}_i\nabla V^*_{i}$. 
According to $\left\| \nabla\epsilon_i  \right\| \leq b_{\epsilon_{ei}}$ in Assumption \ref{bound of NN issues},  $\left\| \epsilon_{\Delta u^*_{i}} \right\| \leq \frac{1}{2} \left\|\bar{g}_i\right\|b_{\epsilon_{ei}}$ holds.
Then, by combining \eqref{optimal incremental hat u} with \eqref{stability dLv 41}, we get 
\begin{equation} \label{stability dLv 5}
\begin{aligned}
\Delta \hat{u}_{i_{r}}-\Delta u^*_{i_{r}} =  \beta(\tanh(\mathscr{G}^*_i)-\tanh(\hat{\mathscr{G}}_i))+\epsilon_{\Delta u^*_{i}}.
\end{aligned}
\end{equation}
where  $\mathscr{G}^*_i = \frac{1}{2\beta}\bar{g}^{\top}_i\nabla \Phi^{\top}_{i}W^*_{i}$, and $\hat{\mathscr{G}}_i = \frac{1}{2\beta}\bar{g}^{\top}_i\nabla \Phi^{\top}_{i}\hat{W}$.
Based on \eqref{optimal incremental u} and \eqref{optimal incremental hat u}, the Taylor series of $\tanh(\mathscr{G}^*_i)$ follows
\small
\begin{equation} \label{Taylor series}
\begin{aligned}
        \tanh(\mathscr{G}^*_i) &= \tanh(\hat{\mathscr{G}}_{i}) +\frac{\partial \tanh(\hat{\mathscr{G}}_{i})}{\partial\hat{\mathscr{G}}_i}(\mathscr{G}^*_i-\hat{\mathscr{G}}_i)+\mathcal{O}((\mathscr{G}^*_i-\hat{\mathscr{G}_i})^2)\\
        &=\tanh(\hat{\mathscr{G}}_{i})-\frac{1}{2 \beta}(1-\tanh^2(\hat{\mathscr{G}}_{i}))  \bar{g}^{\top}_i
        \nabla \Phi^{\top}_{i}\Tilde{W}_{i} \\
        & \; \; \; \; + \mathcal{O}((\mathscr{G}^*_i-\hat{\mathscr{G}}_i)^2),
\end{aligned}
\end{equation}
\normalsize
where $\mathcal{O}((\mathscr{G}^*_i-\hat{\mathscr{G}}_i)^2)$ is a higher order term of the Taylor series. By following \cite[Lemma 1]{yang2016online}, this higher order term is bounded as
    \begin{equation}\label{bound of higher order term}
        \begin{aligned}
            \left\| \mathcal{O}((\mathscr{G}^*_i-\hat{\mathscr{G}}_i)^2)  \right\| & \leq 2+\frac{1}{\beta} \left\|\bar{g}_i\right\| b_{\Phi_{ei}} \left\|\Tilde{W}_{i} \right\|.
        \end{aligned}
    \end{equation}
Based on \eqref{Taylor series}, we rewrite \eqref{stability dLv 5}  as 
\begin{equation}\label{u-u*-a}
\begin{aligned}
        \Delta \hat{u}_{i_{r}}-\Delta u^*_{i_{r}} &= \beta(\tanh(\mathscr{G}^*_i)-\tanh(\hat{\mathscr{G}}_{i}))+\epsilon_{\Delta u^*_{i}}\\
        &=-\frac{1}{2}(1-\tanh^2(\hat{\mathscr{G}}_{i}))  \bar{g}^{\top}_i \nabla \Phi^{\top}_{i}\Tilde{W}_{i}\\
        & \; \; \; \;+ \beta \mathcal{O}((\mathscr{G}^*_i-\hat{\mathscr{G}}_i)^2)+\epsilon_{\Delta u^*_{i}}.
        \end{aligned}
 \end{equation}
Then, by combining \eqref{bound of higher order term} with \eqref{u-u*-a}, and given that $\left\|1-\tanh^2(\hat{\mathscr{G}}_{i})\right\| \leq 2$, \emph{$\left\| \Delta \hat{u}_{i_{r}}-\Delta u^*_{i_{r}} \right\|^2$ in \eqref{stability dLv 4} follows}
\small
    \begin{equation}\label{u-u* abs 2}
        \begin{aligned}
            \left\| \Delta \hat{u}_{i_{r}}-\Delta u^*_{i_{r}} \right\|^2 & \leq 3 \beta^2 \left\| \mathcal{O}((\mathscr{G}^*_i-\hat{\mathscr{G}}_i)^2)\right\|^2 +3\left\| \epsilon_{\Delta u^*_{i}} \right\|^2 \\
             & \; \; \; \;+3\left\| -\frac{1}{2}(1-\tanh^2(\hat{\mathscr{G}}_{i})) \bar{g}^{\top}_i\nabla \Phi^{\top}_{i}\Tilde{W}_{i}\right\|^2 \\
             & \leq 6  \left\|\bar{g}_i\right\|^2 b^2_{\Phi_{ei}} \left\|\Tilde{W}_{i}\right\|^2+ 12\beta^2 + \frac{3}{4} \left\|\bar{g}_i\right\|^2b^2_{\epsilon_{ei}} \\
             & \; \; \; \;+ 12 \beta   \left\|\bar{g}_i\right\| b_{\Phi_{ei}} \left\|\Tilde{W}_{i}\right\|.
        \end{aligned}
\end{equation}
\normalsize
Based on \eqref{optimal V approximation}, \eqref{equalivance eq2},  Assumption \ref{bound of NN issues}, and the fact that $\left\|W^{*}_i\right\| \leq b_{W^{*}_{i}}$,
\emph{$\left\| \tanh^{-1}(\Delta u^{*}_{i_{r}} / \beta) \right\|^2$ in \eqref{stability dLv 4} follows}
\small
\begin{equation} \label{stability dLv 51}
\begin{aligned}
\left\| \tanh^{-1}\left(\Delta u^{*}_{i_{r}}/\beta\right) \right\|^2 &= \left\| \frac{1}{4\beta^2} \nabla {V_i^*}^{\top}\bar{g}_i \bar{g}^{\top}_i\nabla V^{*}_{i} \right\| \\
& \leq \frac{1}{4\beta^2} \left\|\bar{g}_i\right\|^2 b^2_{\Phi_{ei}} b^2_{W^{*}_{i}} + \frac{1}{4\beta^2} b^2_{\epsilon_{ei}}\left\|\bar{g}_i\right\|^2 \\
& \; \; \; \; +\frac{1}{2\beta^2}\left\|\bar{g}_i\right\|^2 b_{\Phi_{ei}} b_{\epsilon_{ei}} b_{W^{*}_{i}}.
\end{aligned}
\end{equation}
\normalsize
Using \eqref{u-u* abs 2} and \eqref{stability dLv 51}, \eqref{stability dLv 4} reads 
\begin{equation} \label{stability dLv 6}
\begin{aligned}
&-2 \beta \tanh^{-1}\left(\Delta u^{*}_{i_{r}}/\beta\right)( \Delta \hat{u}_{i_{r}}-\Delta u^*_{i_{r}})
 \leq  \frac{1}{4} \left\|\bar{g}_i\right\|^2 b^2_{\Phi_{ei}} b^2_{W^{*}_{i}} \\
 & + \frac{1}{4} b^2_{\epsilon_{ei}}\left\|\bar{g}_i\right\|^2+\frac{1}{2}\left\|\bar{g}_i\right\|^2 b_{\Phi_{ei}} b_{\epsilon_{ei}} b_{W^{*}_{i}}
+ 6  \left\|\bar{g}_i\right\|^2 b^2_{\Phi_{ei}} \left\|\Tilde{W}_{i}\right\|^2 \\
& + 12\beta^2 + \frac{3}{4} \left\|\bar{g}_i\right\|^2b^2_{\epsilon_{ei}}+ 12 \beta   \left\|\bar{g}_i\right\| b_{\Phi_{ei}} \left\|\Tilde{W}_{i}\right\|.
 \end{aligned}
\end{equation}
Substituting \eqref{stability dLv 6} into \eqref{stability dLv 3}, finally the first term $\dot{L}_{i_{1}}$ follows
\small
\begin{equation} \label{stability dLv 7}
\begin{aligned}
\dot{L}_{i_{1}} \leq & -e^{\top}_i Q_i e_i - (\Bar{\xi}^2_{oi}-\xi_i^{\top}\xi_i ) - \left[ \beta \tanh^{-1}\left(\Delta u^{*}_{i_{r}}/\beta\right)+\xi_i \right]^2 \\
     & + \frac{3}{4} \left\|\bar{g}_i\right\|^2 b^2_{\Phi_{ei}} b^2_{W^{*}_{i}} + \frac{3}{4} b^2_{\epsilon_{ei}}\left\|\bar{g}_i\right\|^2+\frac{3}{2}\left\|\bar{g}_i\right\|^2 b_{\Phi_{ei}} b_{\epsilon_{ei}} b_{W^{*}_{i}}\\
     & + 6  \left\|\bar{g}_i\right\|^2 b^2_{\Phi_{ei}} \left\|\Tilde{W}_{i}\right\|^2+ 12\beta^2 + \frac{3}{4} \left\|\bar{g}_i\right\|^2b^2_{\epsilon_{ei}}\\
             & + 12 \beta   \left\|\bar{g}_i\right\| b_{\Phi_{ei}} \left\|\Tilde{W}_{i}\right\|.
\end{aligned}
\end{equation}
\normalsize
\emph{As for the second term $\dot{L}_{W} = \frac{1}{2} \Tilde{W}^{\top}_{i} \Gamma_i^{-1} \Tilde{W}_{i}$}, based on \eqref{w update law} and Theorem 1 in our previous work \cite{li2021model}, it follows
\begin{equation} \label{stability dLw}
\dot{L}_{i_{2}} \leq - \Tilde{W}_{i}^{\top} \mathcal{Y}_i \Tilde{W}_{i} +\Tilde{W}_{i}^{\top}\epsilon_{\Tilde{W}_{i}}.
\end{equation}
where $\mathcal{Y}_i = \sum_{l=1}^{P_i} k_{e_{i}}Y_{i_{l}}Y^{\top}_{i_{l}} \in \mathbb{R}^{N_i \times N_i}$, and $\epsilon_{\Tilde{W}_{i}} =-k_{t_{i}} Y_i\epsilon_{h_{i}}-\sum_{l=1}^{P_i}k_{e_{i}} Y_{i_{l}}\epsilon_{h_{il}} \in \mathbb{R}^{N_i}$. The boundness of $Y_i$ and $\epsilon_{h_{i}}$ results in bounded $\epsilon_{\Tilde{W}_i}$. Thus, there exists $\Bar{\epsilon}_{\Tilde{W}_{i}} \in \mathbb{R}^{+}$ such that $\left\|\epsilon_{\Tilde{W}_{i}}\right\| \leq \Bar{\epsilon}_{\Tilde{W}_{i}}$. According to Assumption~\ref{rank condition}, $\mathcal{Y}_i$ is positive definite. Thus, \eqref{stability dLw} could be rewritten as
            \begin{equation} \label{dVcl 1}
            \begin{aligned}
                \dot{L}_{i_{2}}
                & \leq - \lambda_{\min}(\mathcal{Y}_i) \left\| \Tilde{W}_i \right\|^2- \Bar{\epsilon}_{\Tilde{W}_{i}}\left\| \Tilde{W}_i \right\|.
            \end{aligned}
        \end{equation}
\emph{Finally, as for $\dot{L}_i$}, substituting \eqref{stability dLv 7} and \eqref{stability dLw} into \eqref{Lya function stability}, we get 
\begin{equation} \label{stability final result}
\dot{L}_i \leq  -\mathcal{A}_i
-\mathcal{B}_i \left\|\Tilde{W}_{i}\right\|^2
+\mathcal{C}_i \left\|\Tilde{W}_{i}\right\| + \mathcal{D}_i,
\end{equation}
where $\mathcal{A}_i = e^{\top}_i Q_i e_i + (\Bar{\xi}^2_{oi}-\xi_i^{\top}\xi_i ) + \left[\beta \tanh^{-1}\left(\Delta u^{*}_{i_{r}}/\beta\right)+\xi_i\right]^2$,
$\mathcal{B}_i = \lambda_{\min}(\mathcal{Y}_i)- 6\left\|\bar{g}_i\right\|^2 b^2_{\Phi_{ei}}$,
$\mathcal{C}_i = 12 \beta   \left\|\bar{g}_i\right\| b_{\Phi_{ei}} +\Bar{\epsilon}_{\Tilde{W}_{i}}$,
and $\mathcal{D}_i = \frac{3}{4} \left\|\bar{g}_i\right\|^2 b^2_{\Phi_{ei}} b^2_{W^{*}_{i}} + \frac{3}{2} b^2_{\epsilon_{ei}}\left\|\bar{g}_i\right\|^2+ \frac{3}{2}\left\|\bar{g}_i\right\|^2 b_{\Phi_{ei}} b_{\epsilon_{ei}} b_{W^{*}_{i}}+ 12\beta^2$.
Let the parameters be chosen such that $\mathcal{B}_i > 0$. Since $\mathcal{A}_i$ is positive definite, the above Lyapunov derivative \eqref{stability final result} is negative if
    \begin{equation} \label{negative condition}
        \begin{aligned}
        \left\| \Tilde{W}_{i} \right\| > \frac{\mathcal{C}_i}{2\mathcal{B}_i}+\sqrt{\frac{\mathcal{C}_i^2}{4\mathcal{B}_i^2}+\frac{\mathcal{D}_i }{\mathcal{B}_i}}.
        \end{aligned}
    \end{equation}
Thus, the weight learning error of the critic agent converges to the residual set
\begin{equation} \label{compact set}
    \Tilde{\Omega}_{\Tilde{W}_{i}} = \left\{\Tilde{W}_{i} | \left\| \Tilde{W}_{i} \right\| \leq \frac{\mathcal{C}_i}{2\mathcal{B}_i}+\sqrt{\frac{\mathcal{C}_i^2}{4\mathcal{B}_i^2}+\frac{\mathcal{D}_i }{\mathcal{B}_i}} \right\}.
\end{equation}
This completes the proof.
\end{proof}

\section{Proof of Lemma 1} \label{app error bound}
\begin{proof}
Combing (6) with (7), the estimation error for the $i$th subsystem (8) follows 
\begin{equation} \label{error bound eq1}
\begin{aligned}
\xi_i & = h_i - \hat{h}_i = h_i - h_{i,0} \\
& = (\bar{g}^{-1}_i-g^{-1}_{i})\Delta \dot{x}_{i}+ (g^{-1}_{i,0}-g^{-1}_{i}) \dot{x}_{i,0}\\
&+g^{-1}(f_i - f_{i,0})+(g^{-1}_{i}-g^{-1}_{i,0})f_{i,0},
\end{aligned}
\end{equation}
where $\Delta \dot{x}_{i} = \dot{x}_{i}-\dot{x}_{i,0}$. Based on (2), (8) and (9), 
an equivalent form of $\Delta \dot{x}_{i}$ follows
\begin{equation} \label{error bound eq2}
\begin{aligned}
    \Delta \dot{x}_{i} & = f_i+g_i u_i - f_{i,0}-g_{i,0} u_{i,0} \\
    & = g_i \Delta u_i+ (g_i - g_{i,0}) u_{i,0} + f_i - f_{i,0} \\
    & = g_i (\Delta u_{i_{b}}+\Delta u_{i_{r}})+ (g_i - g_{i,0}) u_{i,0} + f_i - f_{i,0}.
    \end{aligned}
\end{equation}
Substituting  \eqref{error bound eq2} into \eqref{error bound eq1}, we get
\begin{equation} \label{error bound eq3}
    \begin{aligned}
       \xi_i & = (g_i\bar{g}^{-1}_i - 1) \Delta u_{i_{b}} + (g_i\bar{g}^{-1}_i - 1)\Delta u_{i_{r}} + \delta_{1i},
    \end{aligned}
\end{equation}
where $\delta_{1i} = \bar{g}^{-1}_i (g_i-g_{i,0})u_0+\bar{g}^{-1}_i(f_i - f_{i,0})$.

In the following, we implement the incremental base policy as (11) in favor of theoretical completeness.
For simplicity, denoting $\mu_i = \dot{x}_{d_{i}}-k_{i}e_{i} \in \mathbb{R}^{n_i}$. According to  (7) and (11),  $\Delta u_{i_{b}}$ in \eqref{error bound eq3} follows
\begin{equation} \label{error bound eq4}
    \begin{aligned} 
    \Delta u_{i_{b}} &= \bar{g}^{-1}_i(\mu_i - \bar{g}_i h_{i,0}-\bar{g}_i u_{i,0})\\
    & = \bar{g}^{-1}_i \mu_i- (\bar{g}^{-1}_{i} - g^{-1}_{i,0})\dot{x}_{i,0} - g^{-1}_{i,0} f_{i,0} - u_{i,0}\\
    & = \bar{g}^{-1}_i \mu_i-(\bar{g}^{-1}_{i} - g^{-1}_{i,0}) (f_{i,0}+g_{i,0}u_{i,0})- g^{-1}_{i,0} f_{i,0} - u_{i,0} \\
    & = \bar{g}^{-1}_i \mu_i - \bar{g}^{-1}_{i} (f_{i,0}+g_{i,0}u_{i,0}) \\
    & = \bar{g}^{-1}_i(\mu_i - \mu_{i,0})-\bar{g}^{-1}_i(\dot{x}_{i,0} -\mu_{i,0}),
    \end{aligned}
\end{equation}
where $\mu_{i,0} = \dot{x}_{d_{i,0}}-k_{i}e_{i,0}$.
Besides, combing (8) with  (9) , we get
\begin{equation} \label{error bound eq5}
    \begin{aligned} 
    \dot{x}_i &=  \dot{x}_{i,0} + \bar{g}_i(\Delta u_{i_{b}}+\Delta u_{i_{r}}) + \bar{g}_i \xi_i \\
    & = \dot{x}_{i,0} + \bar{g}_i \bar{g}^{-1}_i(\mu_i-\dot{x}_{i,0}) + \bar{g}_i \Delta u_{i_{r}} + \bar{g}_i \xi_i \\
    & =  \mu_i + \bar{g}_i  \Delta u_{i_{r}} + \bar{g}_i \xi_i.
    \end{aligned}
\end{equation}

Based on the result shown in \eqref{error bound eq5}, we get
\begin{equation} \label{error bound eq6}
    \begin{aligned} 
    \xi_i = \bar{g}^{-1}_i ( \dot{x}_i - \mu_i - \bar{g}_i  \Delta u_{i_{r}}).
    \end{aligned}
\end{equation}
 Accordingly, the following equation establishes
 \begin{equation} \label{error bound eq7}
    \begin{aligned} 
    \xi_{i,0} = \bar{g}^{-1}_i ( \dot{x}_{i,0}-\mu_{i,0}- \bar{g}_i  \Delta u_{i_{b,0}}).
    \end{aligned}
\end{equation}
Based on the result given in \eqref{error bound eq7}, \eqref{error bound eq4} is rewritten as
 \begin{equation} \label{error bound eq8}
    \begin{aligned}
    \Delta u_{i_{b}} &= \bar{g}^{-1}_i(\mu_i - \mu_{i,0}) - \bar{g}^{-1}_i(\dot{x}_{i,0} -\mu_{i,0} \\
     & - \bar{g}_i  \Delta u_{i_{b,0}}) -\Delta u_{i_{b,0}} \\
     & = \bar{g}^{-1}_i(\mu_i - \mu_{i,0}) - \xi_{i,0} - \Delta u_{i_{b,0}}.
    \end{aligned}
\end{equation}

Substituting \eqref{error bound eq8} into \eqref{error bound eq3} yields
 \begin{equation} \label{error bound eq9}
    \begin{aligned}
    \xi_i  = & (1-g_i\bar{g}^{-1}_i)\xi_{i,0}+(1-g_i\bar{g}^{-1}_i)\bar{g}^{-1}_i (\mu_{i,0}-\mu_i)\\
    & +(1-g_i\bar{g}^{-1}_i)(\Delta u_{i_{b,0}} -\Delta u_{i_{r}} ) + \delta_{1i}.
    \end{aligned}
\end{equation}

 In discrete-time domain, \eqref{error bound eq9} can be represented as
  \begin{equation} \label{error bound eq10}
    \begin{aligned}
    \xi_i(k)  = & (1-g_i (k) \bar{g}^{-1}_i)\xi_{i}(k-1)+(1-g_i (k) \bar{g}^{-1}_i) \Delta \tilde{u}_{i_{b}} \\
    & + \delta_{1i} + \delta_{2i},
    \end{aligned}
\end{equation}
where $\Delta \tilde{u}_{i_{r}} = \Delta u_{i_{r}} (k-1) -\Delta u_{i_{r}}(k)$, $\delta_{2i} = (1-g_i (k) \bar{g}^{-1}_i) \bar{g}^{-1}_i (\mu_i(k-1)-\mu_i(k))$.

The constrained input $\left\| \Delta u_{i_{r}}(k) \right\| \leq \beta $ implies that the following equation holds
 \begin{equation} \label{error bound eq14}
    \begin{aligned}
    \left\| \Delta \tilde{u}_{i_{r}} \right\|   \leq  \left\| \Delta u_{i_{r}} (k-1) \right\| + \left\| \Delta u_{i_{r}} (k)  \right\| \leq 2 \beta.
    \end{aligned}
\end{equation}
We choose the value of $\bar{g}_i$ to meet $\left\| 1-g_i (k) \bar{g}^{-1}_i \right\|   \leq \iota_i < 1$, where $\iota_i \in \mathbb{R}^+$. Under a sufficiently high sampling rate, it is reasonable to assume that there exists $\bar{\delta}_{1i}, \bar{\delta}_{2i} \in \mathbb{R}^+$ such that  $ \left\| \delta_{1i} \right\| \leq   \bar{\delta}_{1i}$ , and $\left\| \delta_{2i} \right\| \leq \iota_i \bar{\delta}_{2i}$. Then, the following equations hold:
 \begin{equation} \label{error bound eq15}
    \begin{aligned}
    \left\| \xi_i(k) \right\| & \leq \iota_i \left\| \xi_i(k-1) \right\| +   \iota_i \left\| \Delta \tilde{u}_{i_{r}} \right\| + \bar{\delta}_{1i} + \iota_i \bar{\delta}_{2i} \\
    & \leq \iota^{2}_i \left\| \xi_i(k-2) \right\| + (\iota^{2}_i +\iota_i )\left\| \Delta \tilde{u}_{i_{r}} \right\|  \\
    & \; \; \; \; + (\iota_i +1)(\bar{\delta}_{1i} + \iota_i \bar{\delta}_{2i})\\
    & \leq \cdots \\
    & \leq \iota^{k}_i \left\|\xi_i(0)  \right\| + \frac{\bar{\delta}_{1i} + \iota_i \bar{\delta}_{2i}}{1-\iota_i} + \frac{ \iota_i \left\| \Delta \tilde{u}_{i_{r}} \right\|}{1-\iota_i} \\
    & \leq \iota^{k}_i \left\|\xi_i(0)  \right\| + \frac{\bar{\delta}_{1i} + \iota_i \bar{\delta}_{2i}}{1-\iota_i} + \frac{2\iota_i \beta}{1-\iota_i} := \bar{\xi}_i.
    \end{aligned}
\end{equation}
    As $k \to \infty $, $ \bar{\xi}_i  \to \frac{\bar{\delta}_{1i} + \iota_i \bar{\delta}_{2i}}{1-\iota_i} + \frac{2\iota_i \beta}{1-\iota_i}$.
    This concludes the proof.
\end{proof}

\section{Settings for 7-DoF  Robot Manipulator} \label{sec kuka}
The parameter settings for the joint-space tracking task of the 7-DoF KUKA iiwa robot manipulator is presented in Table~\ref{tab parameter setting 7DoF joint-space}. The associated simulation results of the 7 subsystems are presented in Fig.~\ref{kuka traj of w and e sum}, which fully validate the superiority of our proposed DR-RL method.

\begin{table}[ht] 
  \begin{center}
    \caption{The parameter settings for KUKA iiwa.}
    \label{tab parameter setting 7DoF joint-space}
    \resizebox{0.5\textwidth}{!}{
    \begin{tabular}{c|c|c}
    \hline
       & \textbf{Incremental Base Policy (IBP)}  & \textbf{DR-RL} \\ \hline
      \multirow{1}{*}{\textbf{Initial value}} 
      & $q(0) = [0,0,0,0,0,0,0]^{\top}$, & $q(0) = [0,0,0,0,0,0,0]^{\top}$, \\
      \multirow{1}{*}{\textbf{conditions}} 
            & $\dot{q}(0) = [0,0,0,0,0,0,0]^{\top}$, & $\dot{q}(0) = [0,0,0,0,0,0,0]^{\top}$, \\
      \multirow{2}{*}{} 
            &$u(0) = [0,0,0,0,0,0,0]^{\top}$, & $u(0) = [0,0,0,0,0,0,0]^{\top}$, \\ 
      & &  $\hat{W}_i(0)= 0_{4 \times 1}$,$i = 1,2,\cdots,7$.\\ 
      \hline
      \multirow{6}{*}{\textbf{parameters}} 
      & $\Bar{g}_1 = 200$, $\Bar{g}_2 = 80$, $\Bar{g}_3 = 100$, $\Bar{g}_4 = 800$, 
      & $\Bar{g}_1 = 200$, $\Bar{g}_2 = 80$, $\Bar{g}_3 = 100$, $\Bar{g}_4 = 800$, \\
      & $\Bar{g}_5 = 800$, $\Bar{g}_6 = 800$, $\Bar{g}_7 = 800$, 
      & $\Bar{g}_5 = 800$, $\Bar{g}_6 = 800$, $\Bar{g}_7 = 800$, \\
      & $k_{1} = \diag{[1,1]}$, $k_{2} = \diag{[0.1,0.1]}$,
      & $k_{1} = \diag{[1,1]}$, $k_{2} = \diag{[0.1,0.1]}$, \\
      & $k_{3} = \diag{[1,1]}$, $k_{4} = \diag{[0.1,0.1]}$,
      & $k_{3} = \diag{[1,1]}$, $k_{4} = \diag{[0.1,0.1]}$, \\
      & $k_{5} = \diag{[1,1]}$, $k_{6} = \diag{[0.1,0.1]}$,
      & $k_{5} = \diag{[1,1]}$, $k_{6} = \diag{[0.1,0.1]}$, \\
      & $k_{7} = \diag{[1,1]}$,  
      & $k_{7} = \diag{[1,1]}$,   \\ 
     &$\Gamma_1 = \diag{[10,10,10,10]}$, $\Gamma_2 = \diag{[1,1,1,1]}$
     & $\Gamma_1 = \diag{[10,10,10,10]}$, $\Gamma_2 = \diag{[1,1,1,1]}$,\\
     &$\Gamma_3 = \diag{[10^3,10^3,10^3,10^3]}$, $\Gamma_4 = \diag{[1,1,1,1]}$,
     & $\Gamma_3 = \diag{[10^3,10^3,10^3,10^3]}$, $\Gamma_4 = \diag{[1,1,1,1]}$,\\
     &$\Gamma_5 = \diag{[10,10,10,10]}$, $\Gamma_6 = \diag{[10,10,10,10]}$,
     & $\Gamma_5 = \diag{[10,10,10,10]}$, $\Gamma_6 = \diag{[10,10,10,10]}$,\\
     &$\Gamma_7 = \diag{[10^7,10^7,10^7,10^7]}$,
     & $\Gamma_7 = \diag{[10^7,10^7,10^7,10^7]}$, \\
     &  & $Q_1 = \diag{[1,1]}$, $Q_2 = \diag{[2000,2000]}$, \\
     &  & $Q_3 = \diag{[10,10]}$, $Q_4 = \diag{[20,20]}$,\\
     &  & $Q_5 = \diag{[1,1]}$, $Q_6 = \diag{[20,20]}$,\\
     &  & $Q_7 = \diag{[1,1]}$, \\ 
      &  & $\beta = 1$,  $P_i = 10$, $\bar{c}_i = 2$, $k_{t_{i}} =1$, $k_{e_{i}} =1$, $i = 1,2,\cdots,7$. \\
      \hline
    \end{tabular}}
  \end{center}
\end{table}

 \begin{figure*}[!t]
    \centering
    \includegraphics[width=7.2in]{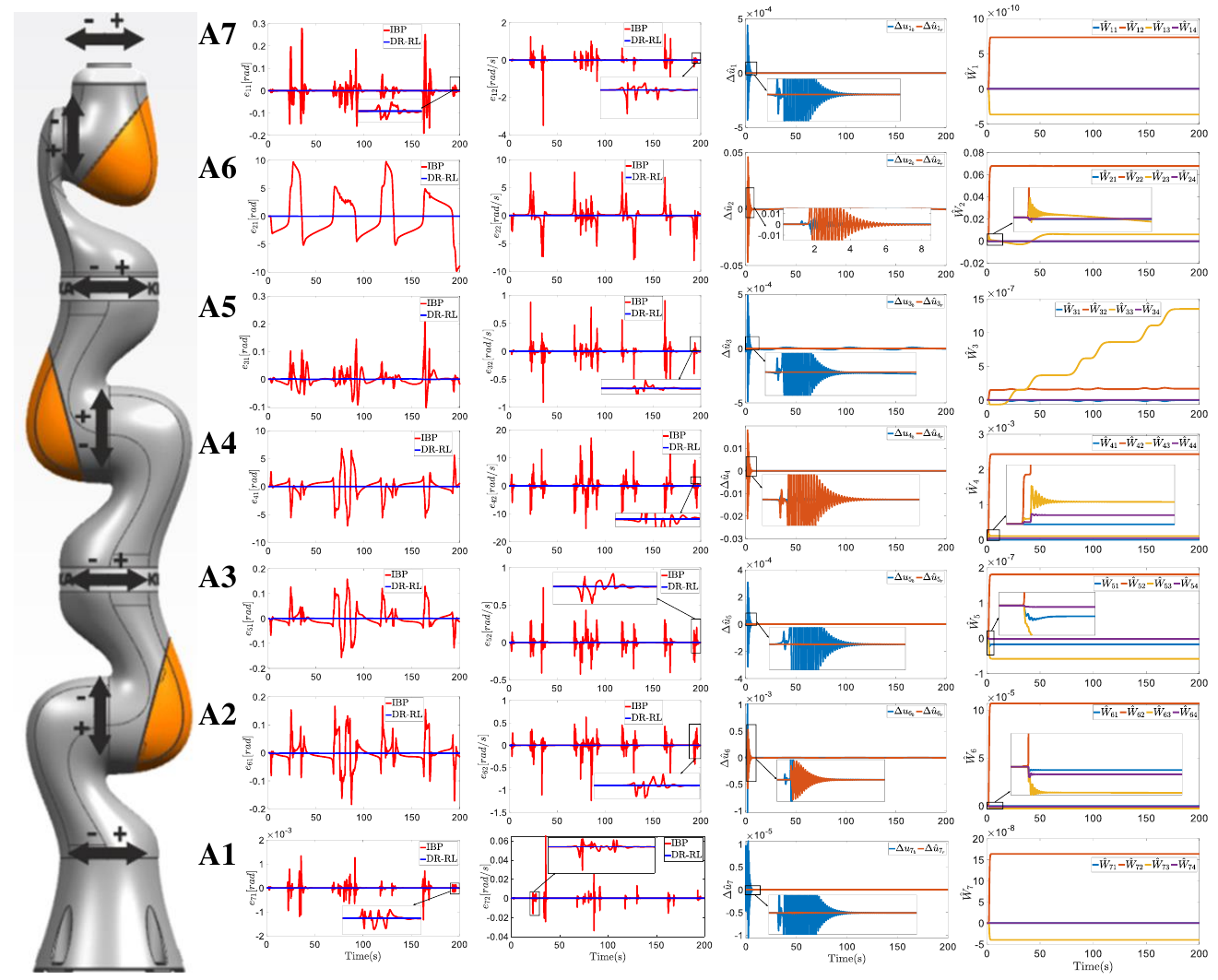}
    \caption{
    The numerical validation results on the 7-DoF KUKA iiwa robot manipulator.
    The simulation results of the subsystem 1 to the subsystem 7 are displayed from bottom to top.
    In each row, each subsystem's evolution trajectories of angle error $e_{i1}$, angle velocity error $e_{i2}$, incremental control inputs $\Delta u_{i_{b}}$ and $\Delta u_{i_{r}}$, and weight $\hat{W}_{i}$ are presented from left to right in sequence.}
    \label{kuka traj of w and e sum}
\end{figure*}

\section{Settings for 6-DoF  Quadrotor} \label{sim 6dof uav}
This section further certifies the effectiveness of DR-RL based tracking control policy under a high-dimensional quadrotor tracking task. 
The quadrotor is driven to track the desired spiral reference trajectory $x_r = [\frac{3}{10\sin(t)},\cos(t),\frac{t}{10\pi},0]^{\top} \in \mathbb{R}^3$, $t \in \left[0,50\right]$.
The associated parameter settings to conduct numerical simulations are referred to Table \ref{tab parameter setting uav}.
The detailed procedures to decouple the 6-DoF quadrotor into 6 subsystems (2) are referred to Section \ref{app quadrotor}.
For subsystem 1-6, we adopt the same activation functions as the ones used for the KUKA iiwa robot manipulator case. 
As displayed in Fig.~\ref{traj of e uav}, we obtain a satisfying tracking performance via the DR-RL based tracking control policy.
\begin{table}[ht]
  \begin{center}
    \caption{The parameter settings of a quadrotor task.}
    \label{tab parameter setting uav}
    \resizebox{0.48\textwidth}{!}{
    \begin{tabular}{c|c}
    \hline
      \multirow{1}{*}{\textbf{Initial value}} 
       &$\xi(0) = [0.1,1.1,0]^{\top}$, $\eta(0) = [0,0,0]^{\top}$,\\
      \multirow{1}{*}{\textbf{conditions}} 
       & $u(0) = [0,0,0.5]^{\top}$, \\ 
       \multirow{2}{*}{} 
      & $\bar{g}_i = 300$, $i = 1,2,3$; $\bar{g}_i = 60000$, $i = 4,5,6$, \\
      & $k_i = \diag{[3,3]}$, $\hat{W}_i(0)= 0_{4 \times 1}$, $i = 1, \cdots,6$. \\ \hline
      \multirow{1}{*}{\textbf{Cost function}} 
       & $Q_i = \diag{(1,1)}$, $\bar{c}_i = 4$,\\
       \multirow{1}{*}{\textbf{parameters}} 
       & $\beta = 0.1$, $i = 1, \cdots,6$. \\ \hline
      \multirow{1}{*}{\textbf{Weight learning}} 
       &$k_{t_{i}} = 1$, $k_{e_{i}} = 0.01$, $P_i = 6$,\\
       \multirow{1}{*}{\textbf{parameters}} 
      & $\Gamma_{i} = 0.01\diag{(I_{1 \times 4})}$, $i = 1, \cdots,6$. \\
      \hline
    \end{tabular}
    }
  \end{center}
\end{table}

\begin{figure}[!t]
    \centering
    \includegraphics[width=3.6in]{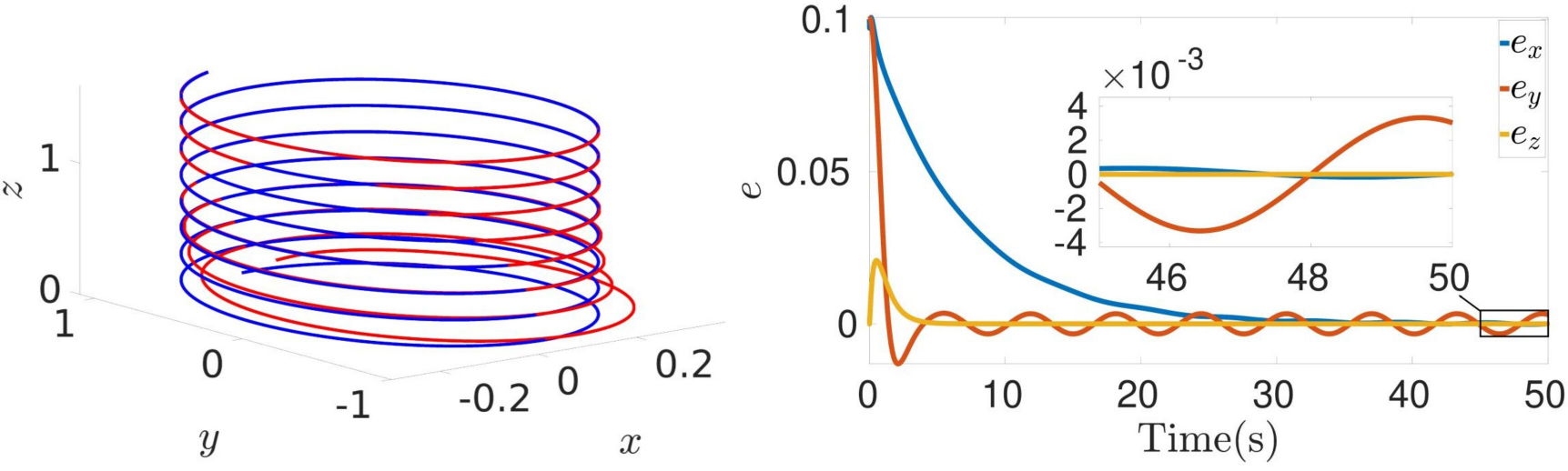}
    \caption{
   The numerical validations about a 6-DoF quadrotor.
   Left: the quadrotor trajectory in 3-D space;
   Right: the evolution trajectories of the position tracking error.
    } 
    \label{traj of e uav}
\end{figure}

\section{Settings for 2-DoF  Robot Manipulator}  \label{sim setting}
The  2-DoF robot manipulator's dynamics follows [18]
\begin{equation*} \label{2DoF model}
M(q) \Ddot{q}+C(q,\dot{q})\dot{q}+F_d \dot{q} + F_s  =\tau,
\end{equation*}
where  $q = [q_1, q_2]^{\top}, \dot{q} = [\dot{q}_1, \dot{q}_2]^{\top},  \tau \in \mathbb{R}^2$; 
Let  $c_{2} = \cos q_{2}$,  $s_{2} = \sin q_{2}$, then $\small{M(q) = \begin{bmatrix}
    a_{1}+2a_{3}c_2 & a_{2}+a_{3}c_2 \\
    a_{2}+a_{3}c_2 & a_{2} \\
\end{bmatrix}}\in \mathbb{R}^{2 \times 2}$, and $\small{C(q,\dot{q}) = \begin{bmatrix}
    -a_{3}\dot{q}_{2}s_2 & -a_{3}(\dot{q}_{1}+\dot{q}_{2})s_2\\
    a_{3}\dot{q}_{1} s_2 & 0\\
\end{bmatrix}}\in \mathbb{R}^{2 \times 2}$, 
wherein $a_1 = 3.473 \, \si{\kilo\gram \square\metre}$, $a_2 = 0.196 \, \si{\kilo\gram \square\metre}$, $a_3 = 0.242 \, \si{\kilo\gram \square\metre}$; 
The static friction follows $F_d = \diag[5.3, 1.1] \si{\newton \metre \second}$,
and the dynamic friction is $F_s = [8.45 \tanh(\dot{q}_1), 2.35 \tanh(\dot{q}_2)]^{\top} \si{\newton \metre \second}$.

\begin{table}[ht] 
  \begin{center}
    \caption{The parameter settings for the task-space task.}
    \label{tab parameter setting 2DoF task-space}
    \resizebox{0.48\textwidth}{!}{
   \begin{tabular}{c|c|c}
    \hline
       & \textbf{DF-RL} [17] & \textbf{DR-RL} \\ \hline
      \multirow{1}{*}{\textbf{Initial value}} 
      & $x(0) = [-0.18,\frac{\pi}{2},0,0]^{\top}$ & $x(0) = [-0.18,\frac{\pi}{2},0,0]^{\top}$ \\
      \multirow{1}{*}{\textbf{conditions}} 
      &$u(0) = [0,0]^{\top}$ & $u(0) = [0,0]^{\top}$  \\ 
      \multirow{2}{*}{} 
      &$\hat{W}(0)= 0_{10 \times 1}$ &  $\hat{W}_i(0)= 0_{4 \times 1}$ \\ 
      \multirow{1}{*}{} 
      & &$\bar{g}_1 = 3.8$, $\bar{g}_2 = 8.3$ \\ 
           \multirow{1}{*}{} 
      & &$k_i = \diag{[10,10]}$  \\  \hline
      \multirow{2}{*}{\textbf{Cost function}} 
      & $Q_1 = \diag{(10.4,23.4)}$ & $Q_1 = \diag{(21,47)}$ \\
      & $Q_2 = \diag{(15.6,18.2)}$ &$Q_2 = \diag{(30,40)}$ \\ 
      \multirow{1}{*}{\textbf{parameters}} 
      & $R = \diag{(1,1)}$ & $R_i = \diag{(1,1)}$ \\ 
      \multirow{1}{*}{}
      & &$\beta_i = 2$ $\bar{c}_i = 1$ \\\hline
      \multirow{1}{*}{\textbf{Weight learning}} 
      & $k_t = 0.9$, $k_e = 0.08$ & $k_{t_{i}} = 200$,  $k_{e_{i}} = 0.2$ \\
      \multirow{1}{*}{\textbf{parameters}} 
      &$\Gamma = 3.2\diag{(I_{1 \times 10})}$ & $\Gamma_{i} = 0.01\diag{(I_{1 \times 4})}$ \\
      &  $P = 15$ & $P_i = 12$ $i = 1,2$. \\
      \hline
    \end{tabular}
    }
  \end{center}
\end{table}

The  10-D activation function used in for task-space tracking task of the 2-DoF robot manipulator is :
\begin{equation*}
\begin{aligned}
\Phi_i(\sigma_i) =  [&\sigma^2_{i_{1}}, \sigma^2_{i_{2}}, \sigma^2_{i_{3}}, \sigma^2_{i_{4}}, \sigma_{i_{1}} \sigma_{i_{2}}, \\
&\sigma_{i_{1}} \sigma_{i_{3}}, \sigma_1 \sigma_{i_{4}}, \sigma_{i_{2}} \sigma_{i_{3}},\sigma_{i_{2}} \sigma_{i_{4}},
\sigma_{i_{3}} \sigma_{i_{4}}]^{\top},
\end{aligned}
\end{equation*}
where $\sigma_i = [e^{\top}_i(t),e^{\top}_i(t-t_s)]^{\top} \in \mathbb{R}^4$, 
$e_i = x_i-x_{d_i} \in \mathbb{R}^2$.

\section{Comparison between DR-RL and AS-RL Methods}  \label{sim dof 2}
This section focuses on a complex robotic manipulation task composed of trajectories represented as different trajectory dynamics to validate the enhanced task flexibility of our proposed DR-RL method. 
Note that the baseline RL-based tracking control strategy [18] developed under an augmented system  (referred to as AS-RL for simplicity)   would fail on the above task, as illustrated in Fig.~\ref{comparion of error sim}.
The considered task is to design a control input $\tau$ to enable the state $x = [q_1, q_2, \dot{q}_1, \dot{q}_2]^{\top}$ to perfectly follow the desired trajectory $ x_d  = [k_{p_{1}} \cos(t),k_{p_{2}} \cos(t),-k_{p_{1}} \sin(t), -k_{p_{2}} \sin(t)]^{\top}$,
where $k_{p_{1}} = 0.5 $, $k_{p_{2}} = 1$ for $t \in \left [0, \frac{61 \pi}{2}\right )$, and $k_{p_{1}} = 0.25$, $k_{p_{2}} = 0.5$ for $t \in \left [\frac{61 \pi}{2}, 400 \right]$. 
This adopted piecewise trajectory $x_d$ violates Assumption 2 in [18] that the desired trajectory is generated from one assumed reference trajectory dynamics. 

\begin{table}[ht] 
  \begin{center}
    \caption{The parameter settings for the joint-space task.}
    \label{tab parameter setting 2DoF joint-space}
    \resizebox{0.48\textwidth}{!}{
    \begin{tabular}{c|c|c}
    \hline
       & \textbf{AS-RL} [18] & \textbf{DR-RL} \\ \hline
      \multirow{1}{*}{\textbf{Initial value}} 
      & $x(0) = [0.5,1.1,0,0]^{\top}$, & $x(0) = [0.5,1.1,0,0]^{\top}$, \\
      \multirow{1}{*}{\textbf{conditions}} 
      &$u(0) = [0,0]^{\top}$, & $u(0) = [0,0]^{\top}$, $\bar{g}_i = 7$, \\ 
      \multirow{2}{*}{} 
      &$\hat{W}(0)= 0_{23 \times 1}$, &  $\hat{W}_i(0)= 0_{4 \times 1}$, \\ 
      &$x_d (0) = [0.5,1,0,0]^{\top}$. & $k_i = \diag{[10,10]}$. \\ \hline
      \multirow{1}{*}{\textbf{Cost function}} 
      & $Q = \diag{(600,600,1,1)}$, & $Q_i = \diag{(600,1)}$, \\ 
      \multirow{1}{*}{\textbf{parameters}} 
      & $R = \diag{(1,1)}$. & $\beta = 0.1$, $\bar{c}_i = 29$.   \\ \hline
      \multirow{1}{*}{\textbf{Weight learning}} 
      & $k_t = 80$, $k_e = 1$, & $k_{t_{i}} = 80$, $k_{e_{i}} = 1$, \\
      \multirow{1}{*}{\textbf{parameters}} 
      &$\Gamma = \diag{(I_{1 \times 23})}$. & $\Gamma_{i} = \diag{(I_{1 \times 4})}$, \\
      &  $P = 25$, & $P_i = 12$, $i = 1,2$. \\
      \hline
    \end{tabular}
    }
  \end{center}
\end{table}

 \begin{figure}[!t]
 \centering
\includegraphics[width=3.6in]{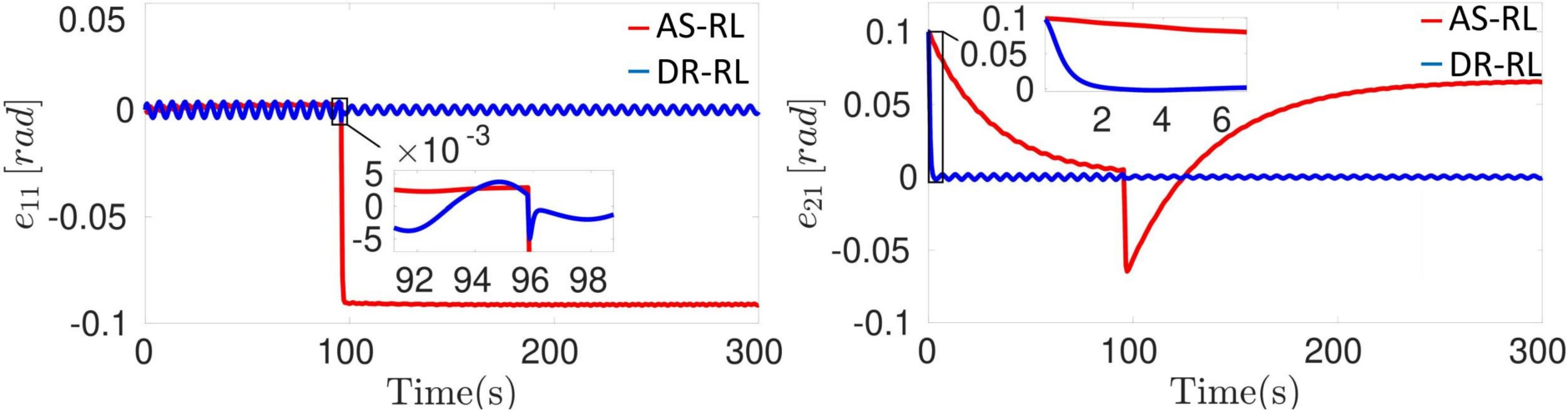}
    \label{joint 2 position error}
    \caption{The evolution trajectories of the $i$th subsystem's the tracking error $e_{i_{1}}$ of AS-RL and DR-RL methods, $i = 1,2$.}
    \label{comparion of error sim}
    \end{figure}
To achieve an accurate value function approximation, we implement the AS-RL method by using the following 23-D activation function provided in [18] but adopting our proposed weight update law (27) to achieve a fair comparison
\footnote{The weight update law proposed in [18] requires directly adding probing noise to control inputs to meet the required persistent excitation condition for the weight convergence, which causes undesirable oscillations.}.
\begin{equation*}\label{basise set 2DoF}
\begin{aligned}
\Phi(\rho) = \frac{1}{2} & [\rho^2_1,\rho^2_2,2\rho_1\rho_3,2\rho_1\rho_4,2\rho_2\rho_3,2\rho_2\rho_4,\rho^2_1\rho^2_2,\rho^2_1\rho^2_5,\\
&\rho^2_1\rho^2_6,\rho^2_1\rho^2_7,\rho^2_1\rho^2_8,\rho^2_2\rho^2_5,\rho^2_2\rho^2_6,\rho^2_2\rho^2_7,\rho^2_2\rho^2_8,\rho^2_3\rho^2_5,\\  
&\rho^2_3\rho^2_6,\rho^2_3\rho^2_7,\rho^2_3\rho^2_8,\rho^2_4\rho^2_5,\rho^2_4\rho^2_6,\rho^2_4\rho^2_7,\rho^2_4\rho^2_8]^{\top},
\end{aligned}
\end{equation*}
where $\rho = [e^{\top},x^{\top}_d]^{\top} \in \mathbb{R}^8$, $e = x-x_d \in \mathbb{R}^4$.
Our developed approach only needs the 4-D activation function same as the one used in Section \ref{sec kuka} for each subsystem.
The detailed parameter settings for AS-RL and DR-RL based tracking control schemes are provided in Table~\ref{tab parameter setting 2DoF joint-space}.

As displayed in Fig.~\ref{comparion of error sim}, both AS-RL and DR-RL based tracking control policies achieve a satisfying performance in the first interval $t \in \left [0, \frac{61 \pi}{2}\right )$. 
However, when the desired trajectory changes at $\frac{61 \pi}{2}$ $\si{\second}$, the learning inefficiency of the AS-RL based tracking control policy (with fine-tuned hyperparameters) hinders us from getting a high-precision tracking performance but with steady-state errors. Our proposed tracking control policy efficiently drives the robot manipulator to track the desired trajectories with 
performance competitive with the model-based AS-RL method, while enjoying better scalability to value function approximation and flexibility to varying desired trajectories. Furthermore, as shown in the right part of Fig.~\ref{comparion of error sim}, the error trajectory of DR-RL enjoys a faster convergence rate than the one of AS-RL, which is due to the guideline from the baseline policy.
Note that the oscillation of the DR-RL method in Fig.~\ref{comparion of error sim}  is the cost of keeping the parameters same as the AS-RL method. 
Even so, regarding the varying tasks, our proposed method under bad hyperparameters shows better performance than the well debugged AS-RL method.
Note that our proposed DR-RL would behave better in other parameter settings.

\section{The Incremental Subsystems  of  Quadrotor} \label{app quadrotor}
This section presents the detailed procedures to decouple the 6-DoF quadrotor into 6 incremental subsystems. 
By introducing pseudo control inputs, we transform the original underactuated quadrotor model into a fully-actuated model.
Thereby, our developed tracking control scheme can be applied to a quadrotor.

Let $\zeta = [x,y,z]^{\top} \in \mathbb{R}^3$, and $\eta = [\phi,\theta,\psi]^{\top} \in \mathbb{R}^3$ represent the absolute linear position and Euler angles defined in the inertial frame, respectively.
The E-L equation of a quadrotor follows 
\begin{subequations} \label{UAV EL}  
\begin{align}
& m \Ddot{\zeta} + m g I_{z} = R T_{B} \label{UAV EL position}\\
& J(\eta) \Ddot{\eta} + C(\eta, \dot{\eta})\dot{\eta} = \tau_{B}, \label{UAV EL angle}
\end{align}
\end{subequations}
where $m \in \mathbb{R}^+$ denotes the mass of the quadrotor; $g \in \mathbb{R}^+ $ is the gravity constant; $I_{z} = [0,0,1]^{\top}$ represents a column vector; $T_B = [0, 0, T]^{\top} \in \mathbb{R}^3$, where $T \in \mathbb{R}$ is the thrust in the direction of the body $z$-axis;
$\tau_{B} = [\tau_{\phi},\tau_{\theta},\tau_{\psi}]^{\top} \in \mathbb{R}^3$ denotes the torques in the direction of the corresponding body frame angles;
$R$, $J(\eta)$, $C(\eta, \dot{\eta})  \in \mathbb{R}^{3 \times 3}$ represent the rotation matrix, Jacobian matrix, and Coriolis term, respectively. 

Expanding the translational dynamics \eqref{UAV EL position} yields
\begin{equation}\label{expansion of first eq}
\begin{aligned}
&\Ddot{x} = \frac{1}{m} T (C_\psi S_\theta C_\phi + S_\psi S_\phi) \\
&\Ddot{y} = \frac{1}{m} T (S_\psi S_\theta C_\phi - C_\psi S_\phi) \\
&\Ddot{z} = -g + \frac{1}{m} T C_\theta C_\phi
\end{aligned}
\end{equation}
where $C_{(\cdot)}$ and $S_{(\cdot)}$ denote $\cos{(\cdot)}$ and $\sin{(\cdot)}$, respectively. 

By introducing pseudo controls $u_1 = T (C_\psi S_\theta C_\phi + S_\psi S_\phi)$, $u_2 =T (S_\psi S_\theta C_\phi - C_\psi S_\phi)$, and $u_3 = T C_\theta C_\phi$, and denoting $x_{11} =x$, $x_{12} =\dot{x}$, $x_{21} =y$, $x_{22} =\dot{y}$, $x_{31} =z$, $x_{32} =\dot{z}$,
we finally decouple the translational dynamics \eqref{UAV EL position} into the following three subsystems
\begin{subequations}\label{3 subsytems position}
\begin{align}
& \dot{x}_{11} = x_{12}, \, \, \dot{x}_{12} = \frac{1}{m} u_1 \label{subsys 1}\\ 
& \dot{x}_{21} = x_{22}, \, \, \dot{x}_{22} = \frac{1}{m} u_2 \label{subsys 2} \\ 
& \dot{x}_{31} = x_{32}, \, \,\dot{x}_{32} = -g+\frac{1}{m} u_3.  \label{subsys 3}
\end{align}
\end{subequations}

By following the same procedures (2)--(8) clarified in Section III, we get three subsystems for the rotational dynamics \eqref{UAV EL angle}:
\begin{subequations}\label{3 subsytems attitude}
\begin{align}
& \dot{x}_{41} = x_{42}, \, \, \dot{x}_{42} = -\frac{h_1}{J_{11}}+ \frac{1}{J_{11}} u_4 \label{subsys 4} \\ 
& \dot{x}_{51} = x_{52}, \, \, \dot{x}_{52} = -\frac{h_2}{J_{22}}+ \frac{1}{J_{22}}  u_5 \label{subsys 5} \\ 
& \dot{x}_{61} = x_{62}, \, \,\dot{x}_{32} = -\frac{h_3}{J_{33}}+ \frac{1}{J_{33}}  u_6, \label{subsys 6}
\end{align}
\end{subequations}
where $h_i =\sum_{j=1,j\ne i}^{3} J_{ij}\Ddot{\eta}_{j}+C_i \dot{\eta}_j \in \mathbb{R}$, $i = 1,2,3$;
$u_4 = \tau_\phi$, $u_5 = \tau_{\theta}$, and $u_6 =\tau_{\psi}$.

The aforementioned procedures \eqref{expansion of first eq}-\eqref{3 subsytems attitude} allow us to get 6 subsystems in the same form as (2). Then, we could adopt our developed DR-RL based tracking control policy, as clarified in Section IV, to drive the quadrotor \eqref{UAV EL} to track the predefined reference trajectory $x_r = [x_d,y_d,z_d, \psi_d]^{\top}$. 
Note that after the explicit values of pseudo controls $u_1$, $u_2$, and $u_3$ are gotten, the trust $T$, and reference angles $\phi_d$, $\theta_d$ are obtained as 
\begin{align}
& T = \sqrt{u^2_1 + u^2_2 + u^2_3} \\
& \phi_d = \arctan(\frac{u_1 S_\psi - u_2 C_\psi}{\sqrt{(u_1 C_\psi + u_2 S_\psi)^2+u^2_3}}), \, \, \, \phi_d \in (-\frac{\pi}{2},\frac{\pi}{2}) \\
& \theta_d =  \arctan (\frac{u_1 C_\psi + u_2 S_\psi}{u_3}), \, \, \,\theta_d \in (-\frac{\pi}{2},\frac{\pi}{2})
\end{align}

\end{document}